\documentclass[a4paper,10pt,journal]{IEEEtran}

\IEEEoverridecommandlockouts
\usepackage{graphicx}
\usepackage{epstopdf}
\DeclareGraphicsExtensions{.eps,.pdf} 
\graphicspath{ {figures/} }

\usepackage{cite}
\usepackage{subfigure}
\usepackage{amssymb,amsmath}
\usepackage{algorithmic}
\usepackage{color}
\usepackage{epstopdf}
\usepackage{multirow}
\usepackage{multicol}
\usepackage{cases}
\usepackage{setspace}
\usepackage{amsthm}
\usepackage{caption}
\usepackage{mathrsfs}
\usepackage{delarray}
\usepackage{tikz}
\usepackage[ruled]{algorithm2e}
\usepackage{balance}
\usepackage[mathscr]{euscript}
 \let\mathscr\relax
\usepackage[scr]{rsfso}

\allowdisplaybreaks
\usepackage{bbm}

\newtheorem{remark}{{\bf{Remark}}}

\newtheorem{lemma}{\bf {Lemma}}

\newtheorem{proposition}{{\bf Proposition}}

\newcommand{\st}{{\mathrm{s.t.}}}

\renewcommand{\algorithmicrequire}{\textbf{Input:}}

     \usepackage[compact]{titlesec}
    \titlespacing{\section}{0pt}{2ex}{1ex}
    \titlespacing{\subsection}{0pt}{1ex}{0ex}
    \titlespacing{\subsubsection}{0pt}{0.5ex}{0ex}

\usepackage{capt-of}
\usepackage{dblfloatfix}
\usepackage{varwidth}
\usepackage{tikz}

\usepackage{setspace}
\renewcommand{\baselinestretch}{1} 

\usepackage[T1]{fontenc}
\usepackage[latin9]{inputenc}
\usepackage{geometry}
\usepackage{array}
\usepackage{booktabs}
\usepackage{siunitx}

\usepackage{babel}
\usepackage{hyperref}

\newcommand{\abs}[1]{\left|#1\right|} 
\allowdisplaybreaks

\begin{document}
\markboth{IEEE Internet of Things Journal}{Vo \MakeLowercase{\textit{et al.}}: Joint Service Placement and Resource Optimization in Hierarchical Edge-Cloud Networks}

\title{\huge Joint Service Placement and Resource Optimization in Hierarchical Edge-Cloud Networks \vspace{-5pt}}
\author{
\IEEEauthorblockN{Vo Phi Son, Van-Dinh Nguyen, Minh-Tuong Nguyen, Tuan-Vu Truong, \\ Toan D. Gian, 
 Dinh Thai Hoang, Diep N. Nguyen and Symeon Chatzinotas \\}
 \thanks{This work is supported by the project VUNI.GREEN-X.RISE.AY25-27.03.}
\thanks{V. P. Son,  M.-T. Nguyen, T.-V. Truong and  T. D. Gian are with the Smart Green Transformation Center (GREEN-X) and College of Engineering and Computer Science, VinUniversity, Hanoi 100000, Vietnam  (e-mail: \{son.vp,  tuong.nm, vu.tt, toan.gd\}@vinuni.edu.vn).}
\thanks{V.-D. Nguyen (corresponding author) is with the School of Computer Science and Statistics, Trinity College Dublin, Dublin 2, D02PN40, Ireland (e-mail: dinh.nguyen@tcd.ie).}
\thanks{D. T. Hoang and D. N. Nguyen are with the School of Electrical and Data Engineering, University of Technology Sydney, Sydney, NSW 2007, Australia (e-mails: \{hoang.dinh, diep.nguyen\}@uts.edu.au).}
\thanks{S. Chatzinotas is  with the Interdisciplinary Centre for Security, Reliability and Trust (SnT), University of Luxembourg (e-mail: symeon.chatzinotas@uni.lu).}
\vspace{-5pt}}

\maketitle

\begin{abstract}
Hierarchical edge-cloud computing-aided Internet of Things (IoT) networks offer low-latency and cost-efficient services to a growing number of data-intensive IoT devices. However, optimizing service placement, which involves determining the most suitable locations within a network to deploy various services, is critical to balancing workloads dynamically and ensuring efficient resource utilization. In this paper, we jointly optimize service placement, edge/cloud cooperation, task offloading, and bandwidth allocation to enhance processing efficiency and response times. The main objective is to minimize both the overall end-to-end latency and the system cost, including service deployment and operational costs. The formulated problem belongs to the class of non-convex mixed-integer nonlinear programming, where finding a feasible solution is already challenging. Towards a stable system, we first transform the original problem into a more tractable form and then decompose it into sub-problems which are solved at different timescales.
Combining tools from relaxation and the successive convex approximation method, we develop iterative algorithms to solve these problems efficiently. With an appropriate penalty parameter, the proposed algorithms guarantee convergence to at least a local optimum. We produce extensive numerical results to demonstrate the superior performance of the proposed algorithms over benchmark schemes as well as emphasize the significance of the joint service placement and resource allocation in enhancing system performance and efficiency.
\end{abstract}

\section{Introduction}\label{sec_Intro}
\IEEEPARstart{H}{ierarchical} edge-cloud Internet of Things (IoT) networks are increasingly prevalent in industries and services such as autonomous devices, intelligent transportation, healthcare, and augmented reality games \cite{Abbas-2018,huynh_latency,Wang_2020}. Sensors, robots, and devices communicate with access points (APs) over dynamic wireless channels, relying on edge-cloud networks for data transmission and computation. These networks must meet demands for low latency, real-time connectivity, and reliability, posing significant quality of service (QoS) challenges \cite{Vaezi_2022,EC_survey}. Hierarchical edge-cloud computing (HECC) addresses these issues by dividing computing tasks into two tiers \cite{EC_survey,Wang_2020,Shah-Mansouri_2018}. In particular, edge servers handle low-latency tasks, while cloud servers process larger workloads, optimized through joint allocation of radio resources, user association, and computation \cite{cooperation2017}. 

Task offloading plays a central role in HECC, allowing users to delegate portions of tasks to edge or cloud servers to reduce execution delays and energy consumption \cite{Wang_latency_taskoff}. However, the complexity of tasks and the heterogeneous, limited computational resources of edge servers (ESs) complicate offloading decisions \cite{Li_two_timescales}. Additionally, large data transmissions to servers can create bottlenecks \cite{huynh_latency}, while insufficient bandwidth allocation exacerbates transmission delays and increases end-to-end (e2e) latency. To address these challenges, jointly optimizing offloading strategies and resource allocation has proven to be an effective approach \cite{Liu_resource_frame}.

Building on effective task offloading, service placement in HECC plays a key role in optimizing latency, energy consumption, and network costs, including software installation, maintenance, and storage \cite{Han_2022,Chen_2021}. While the central cloud can handle most tasks, its distance leads to higher delays and transfer costs. Deploying services at ESs reduces latency and energy use \cite{Gao_2023}, but ESs have limited resources and cannot host all services. Moreover, user mobility, fluctuating demands, and unstable wireless conditions can overload ESs, and thus compromising QoS \cite{Ma_2020}. The heterogeneity of ESs' resources and varying service request densities further complicate service placement decisions and elevate system costs \cite{Han_2022, Gao_2023, Ouyang_2018}.

\subsection{Motivation}\label{sec_motivation}
To address the multifaceted challenges in HECC-aided IoT systems, the joint optimization of service placement, user association, edge cooperation, task offloading, and bandwidth allocation has become a crucial research focus. Service placement plays a fundamental role in determining where and how services are deployed across the network, directly impacting resource utilization, scalability, and overall system efficiency. By strategically placing services, the number of deployments at ESs can be minimized, reducing network costs while accounting for the limited computational and storage resources available \cite{Zhou_cost, huynh_service_placement, Han_2022, Gao_2023, Ouyang_2018}. In parallel, optimizing radio resource allocation and task offloading proportions significantly reduces task execution latency \cite{huynh_latency}. However, most existing studies primarily focus on basic cost metrics, such as computation and transmission overhead, while overlooking broader system challenges, such as dynamic workload balancing, service migration, and network congestion.

Maximizing the potential of HECC-aided IoT systems requires addressing complex, interconnected challenges. These include ensuring high QoS requirements, managing heterogeneous ES computational resources, mitigating limited radio resources, and adapting to unpredictable user mobility and service demands under unstable wireless conditions\cite{Gao_service_place,Wang_2020,Shah-Mansouri_2018,Ouyang_2018}. Simultaneous task offloading from multiple users can overwhelm ESs, creating bottlenecks that increase processing delays and e2e latency \cite{huynh_latency}. Moreover, increasing operational costs from supporting more services in ESs further complicate system efficiency, especially without optimized user association, edge-edge cooperation, and edge-cloud cooperation.

While the joint optimization of service placement and task offloading can effectively reduce latency, frequent updates on short timescales can significantly increase system costs \cite{Ouyang_2018}. For instance, shorter time slots may require ESs to repeatedly install or remove computationally intensive services, such as video analytics or machine learning, to adapt to dynamic user demands, leading to excessive overhead \cite{Farhadi_2019,Li_two_timescales}. A promising solution is to integrate long-term service placement and user association optimization with short-term task offloading and radio resource allocation \cite{Liang_two_timescales,Wei_two_timescales}. This hierarchical optimization approach enhances efficiency, scalability, and cost-effectiveness in HECC-aided IoT systems by balancing system stability with adaptability to real-time demands.

\subsection{Main Contributions}\label{sec_main_contribution}
We propose a hierarchical optimization framework that integrates long-term service placement and user association with short-term task offloading and radio resource allocation for HECC-aided IoT systems. The proposed framework jointly optimizes system configuration, task offloading, bandwidth allocation, and dynamic user demands to improve ES utilization and QoS. By jointly optimizing these key components, the framework minimizes end-to-end latency and operational costs while leveraging edge-edge and edge-cloud cooperation to avoid service bottlenecks. This scalable and cost-effective solution is well suited for dynamic and resource-constrained IoT environments. The main contributions of this paper are summarized as follows:
\begin{itemize}
    \item We develop a unified optimization framework that jointly optimizes service placement, user association, edge-edge cooperation, edge-cloud cooperation, task offloading, and bandwidth allocation to minimize overall end-to-end latency and system cost in HECC-aided IoT networks. In particular, we design an efficient three-tier task offloading mechanism across users, edge servers, and the cloud by jointly determining the portions of tasks executed locally, at edge servers, or at the cloud, together with system configuration decisions such as user association and service placement. The proposed framework effectively addresses heterogeneous computing resources, ES capacity constraints, and wireless channel uncertainties. In addition, we introduce a system cost metric that captures the monetary cost of service placement while ensuring the required QoS.
    
    \item The formulated problem belongs to the class of nonconvex mixed-integer nonlinear programming (MINLP), where even finding a feasible solution is challenging. By exploiting the inherent structure of the problem, we transform it into a more tractable form through three lemmas on service placement, connectivity-dependent offloading, and edge-edge and edge-cloud cooperation policies. The resulting problem is decomposed into two subproblems operating at different timescales to ensure a stable network configuration. By combining successive convex approximation (SCA) and penalty-based techniques \cite{AmirBeck2010}, we develop efficient iterative algorithms that converge to at least a local optimum.

    \item Extensive numerical results demonstrate the effectiveness of the proposed framework in improving network resource utilization. The results show that the proposed approach consistently outperforms benchmark schemes across multiple performance metrics while effectively reducing system costs.
\end{itemize}
Table~\ref{tab:comparison}  summarizes the key differences between our work and the most relevant prior studies, showing that the proposed framework provides a more comprehensive optimization approach than existing methods.
\begin{table*}[!t]
\centering
\caption{Comparison of Notable Related Works}
\label{tab:comparison}
\footnotesize
\renewcommand{\arraystretch}{1.15}
\setlength{\tabcolsep}{4pt}

\begin{tabular}{clcccccccccccccc}
\toprule
\# & Criterion & Our Work & \cite{huynh_latency} & \cite{Li_two_timescales} &\cite{Liu_resource_frame}  &  \cite{Han_2022} & \cite{Gao_2023} & \cite{Ouyang_2018} & \cite{Gao_service_place}   &   \cite{Liang_two_timescales}  & \cite{ Wei_two_timescales}   & \cite{ Yan_2020} & \cite{Zhang_2021} & \cite{Huynh_IOT_2024} & \cite{Fan_TMC_2024} \\
\midrule

1 & Hierarchical user-edge-cloud computing
& $\checkmark$ & $\checkmark$ & $\times$ & $\times$ & $\times$ & $\times$ & $\times$ & $\times$ & $\times$ & $\checkmark$ & $\times$ & $\checkmark$ & $\times$ & $\checkmark$\\

2 & Joint optimization of latency and system cost
& $\checkmark$ & $\times$ & $\times$ & $\times$ & $\checkmark$ & $\checkmark$ & $\checkmark$ & $\times$ & $\checkmark$ & $\times$ & $\times$ & $\checkmark$ & $\checkmark$ & $\checkmark$\\

3 & Service placement optimization
& $\checkmark$ & $\times$ & $\checkmark$ & $\times$ & $\checkmark$ & $\checkmark$ & $\checkmark$ & $\checkmark$ & $\times$ & $\times$ & $\checkmark$ & $\checkmark$ & $\checkmark$ & $\checkmark$\\

4 & User association optimization 
& $\checkmark$ & $\checkmark$ & $\times$ & $\checkmark$ & $\checkmark$ & $\times$ & $\times$ & $\checkmark$ & $\checkmark$ & $\times$ & $\times$ & $\times$ & $\times$ & $\times$ \\

5 & Edge-edge cooperation
& $\checkmark$ & $\times$ & $\checkmark$ & $\times$ & $\times$ & $\times$ & $\checkmark$ & $\checkmark$ & $\times$ & $\times$ & $\times$ & $\times$ & $\times$ & $\checkmark$ \\

6 & Edge-cloud cooperation
& $\checkmark$ & $\checkmark$ & $\times$ & $\times$ & $\times$ & $\times$ & $\times$ & $\times$ & $\times$ & $\times$ & $\times$ & $\checkmark$ & $\times$ & $\checkmark$\\

7 & Task offloading optimization
& $\checkmark$ & $\checkmark$ & $\checkmark$ & $\checkmark$ & $\times$ & $\checkmark$ & $\times$ & $\times$ & $\checkmark$ & $\checkmark$ & $\checkmark$ & $\checkmark$ & $\checkmark$ & $\checkmark$\\

8 & Bandwidth allocation optimization
& $\checkmark$ & $\times$ & $\times$ & $\times$ & $\times$ & $\times$ & $\times$ & $\times$ & $\times$ & $\times$ & $\times$ & $\times$ & $\checkmark$ & $\times$\\

9 & Cost-aware service placement model
& $\checkmark$ & $\times$ & $\checkmark$ & $\times$ & $\checkmark$ & $\checkmark$ & $\times$ & $\checkmark$ & $\checkmark$ & $\checkmark$ & $\checkmark$ & $\checkmark$ & $\checkmark$ & $\checkmark$\\

10 & MINLP problem
& $\checkmark$ & $\checkmark$ & $\checkmark$ & $\times$ & $\checkmark$ & $\times$ & $\times$ & $\times$ & $\times$ & $\times$ & $\times$ & $\times$ & $\checkmark$ & $\times$\\

\bottomrule
\end{tabular} 
\end{table*}

\subsection{Paper Structure and Notations}\label{sec_paper_structure_notations}
The rest of the paper is organized as follows. Section \ref{sec_Relatedwork} reviews related work, while Section \ref{sec_SystemModel} describes the system model. In Section \ref{sec_Problem}, we formulate the problem and transform it into a more tractable form. Section \ref{sec_Solutions} presents the proposed solution algorithms. Numerical results are given in Section \ref{sec_NumericalResults}, and conclusions are drawn in Section \ref{sec_conclusion}.

\textit{Notation:} Throughout the paper, scalars, vectors, and matrices are denoted by lower-case, boldface lower-case, and boldface upper-case letters, respectively. The modulus of a complex number and the Euclidean norm of a vector are denoted by $\abs{\cdot}$ and $\|\cdot\|$, respectively.

\section{Related Work}\label{sec_Relatedwork}
Task offloading is central to mobile edge computing (MEC) and edge-cloud computing networks for reducing task execution time, conserving mobile user energy, and improving network utility \cite{huynh_latency, Yan_2020, Zhang_2021}. For instance, the authors in  \cite{huynh_latency} developed an alternating optimization-based framework integrating task offloading, resource allocation, user association, and cloud assistance to reduce e2e latency in hierarchical edge-cloud systems. Yan \textit{et al.} \cite{Yan_2020} employed a bi-section search method to address the joint task offloading and resource allocation problem in a task dependency model, aiming to minimize user energy consumption and task processing latency in MEC. Similarly, Zhang \textit{et al.} \cite{Zhang_2021} designed an approximate algorithm leveraging semidefinite relaxation and alternating optimization to minimize delay costs by jointly optimizing task offloading and resource allocation, with service caching for both independent and service-dependent tasks in MEC systems. 

On the other hand, a digital twin approach was introduced in \cite{Huynh_2022}, aiming to minimize worst-case latency in computation offloading for industrial IoT applications through joint user association, task offloading, and resource allocation. Zhao \textit{et al.} \cite{Zhao_taskoff} proposed a novel framework for optimal resource-sharing contracts between edge and cloud servers, jointly optimizing resource-sharing and computation strategies. Lu \textit{et al.} \cite{Lu_taskoff} introduced a lightweight offloading framework for multi-task scenarios to minimize execution time and energy consumption, while Adhikari \cite{Adhikari_task_deadline} demonstrated the efficacy of a priority-aware task offloading framework that uses multi-level feedback queues to meet task deadlines. Lastly, Zhao \textit{et al.} \cite{Zhao_task_UAV_deepReif} developed a multi-agent reinforcement learning framework for UAV systems, optimizing UAV trajectories, computational resources, and radio resource allocation to reduce execution time and energy consumption. Despite their contributions, these studies assume that all services are pre-installed on ESs, leading to higher system costs.

Multi-timescale optimization has emerged as an effective approach for NP-hard problems in edge-cloud computing, addressing the interactions between binary and continuous variables in objective functions and constraints \cite{Liu_resource_frame, Liang_two_timescales, Wei_two_timescales, Li_two_timescales}. For example, a two-timescale framework was developed in \cite{Li_two_timescales} to asynchronously optimize service placement and task offloading, maximizing long-term network utility in MEC systems. Liu \textit{et al.} \cite{Liu_resource_frame} optimized joint long-term user association and short-term dynamic task offloading to improve task reliability and reduce delays. Wei \textit{et al.} \cite{Wei_two_timescales} demonstrated that coordinated resource placement and task offloading improves QoS satisfaction, while Liang \textit{et al.} \cite{Liang_two_timescales} employed a two-timescale algorithm to manage service migration and transmission power for long-term energy savings. Despite these advances, existing two-timescale approaches assume infrequent, hour- to day-level service placement, limiting their ability to react to rapidly changing IoT workloads and network conditions. 

Recently, service placement optimization frameworks have been developed to minimize costs, delays, and computational resource usage in ES networks \cite{Han_2022, Zhang_2022, Gao_2023, Gao_service_place, Ouyang_2018, Chen_service_place, Liu_2022,Huynh_IOT_2024}. Pengchao \textit{et al.} optimized service placement and base station power control to balance delays in MEC networks. The authors in \cite{Zhang_2022} addressed varying service request rates and pricing by jointly optimizing service placement and ES deployment to maximize network profits. A heuristic framework for service placement in satellite-terrestrial edge systems was introduced in\cite{Gao_2023} that minimizes response latency and energy use. Additionally, Gao \textit{et al.} \cite{Gao_service_place} used matching and game theory to jointly optimize service placement and network access, reducing delays and switching costs. The work in \cite{Ouyang_2018} applied a Markov approximation algorithm to dynamic service placement for minimizing costs under user mobility.  However, these studies overlook edge-edge and edge-cloud cooperation to address ES resource constraints, heterogeneous task requirements, and radio limitations, which are crucial for meeting e2e latency and cost objectives.

\section{System Model}\label{sec_SystemModel}
\begin{figure}[t]
	\centering
	\includegraphics[width=0.98\columnwidth,trim={0cm 0.0cm 0cm 0.0cm}]{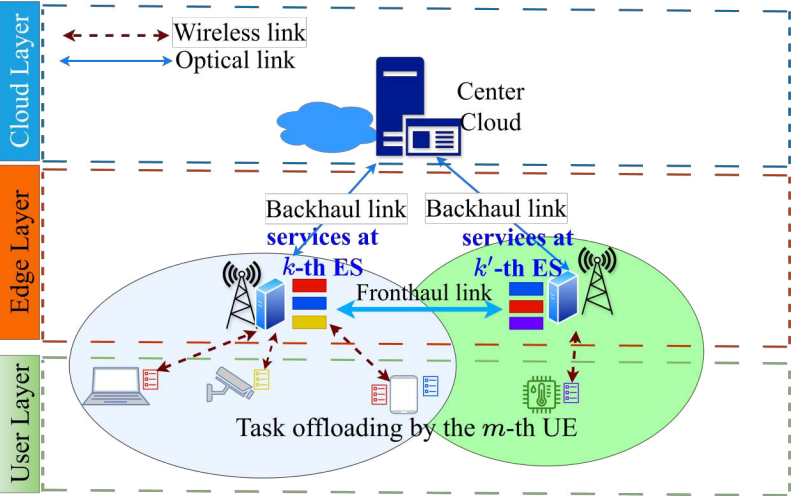}
	\caption{A hierarchical edge-cloud computing systems aided IoT network model.}
	\label{fig:system_model}
\end{figure}
We consider a service-oriented multi-tier computing system (\textit{e.g.} user, edge and cloud layers) for IoT networks, as shown in Fig.~\ref{fig:system_model}. The user equipment (UE) set, denoted as $\mathcal{M} = \{1,2,\dots,M\}$, consists of 
$M$ spatially distributed devices. Each user device, equipped with a single antenna, communicates with a $L$-antenna access point (AP) to offload tasks to the edge layer. The edge layer comprises the set $\mathcal{K} = \{1,2,\dots,K\}$ of $K$ ESs, colocated with APs. Tasks offloaded by UE $m \in \mathcal{M}$ can be processed at the receiving ES $k \in \mathcal{K}$, transferred to a neighboring ES $k' \in \mathcal{K}$, or forwarded to the cloud via fronthaul or backhaul links. This is facilitated by ESs and edge-cloud cooperation for optimized service deployment and placement \cite{cooperation2017, longchen2020,Fan_TMC_2024}.

\subsection{Service Placement and Cooperation Model} \label{sec_Service_placement}
We consider a time-slotted system with long- and short-term timeslots. To ensure system stability, we optimize service placement and server cooperation in the long-term timeslot, while task offloading and resource allocation are optimized at the short-term timeslot to minimize the e2e latency. Let ${\Delta}$ and $\delta$ represent the durations of the long-term timeslot $t$ and short-term timeslot $t_j$, respectively, with ${\delta} = \frac{{\Delta}}{J}$, where $J$ is the number of short-term timeslots within each long-term timeslot.
We denote the set of services as $\mathcal{S} \buildrel \Delta \over = \{1, 2, \dots, S\}$, where $S$ is the total number of services. The service placement decision variable $g_s^k[t] \in \{0, 1\}$ indicates whether the $s$-th service is installed on the $k$-th server at time $t$. Due to resource constraints, each ES can host only a limited number of services, yielding $\sum\nolimits_{s \in \mathcal{S}} {g_s^k\left[ t \right]}  \le S_k^{\max}, \forall k$, with $S_k^{\max} \leq S,\, \forall k. $ 

For task offloading, we introduce new binary variables $\sigma_{m,s}^k[t] \in \{0, 1\}$ to indicate whether UE $m$ offloads tasks for service $s$ to ES $k$ during the long-term timeslot $t$. Specifically, $\sigma_{m,s}^k[t] = 1$ means UE $m$ offloads tasks to ES $k$, and $\sigma_{m,s}^k[t] = 0$, otherwise. We consider that each UE can connect to only one ES at time $t$, such as: $\sum\nolimits_{k\in \mathcal{K}} \sigma_{m,s}^k\left[t\right] \leq 1, \forall m$. The binary variable $\sigma_{m,s}^{k,k'}[t] \in \{0, 1\}$ represents edge-edge cooperation. $\sigma_{m,s}^{k,k'}[t] = 1$ means ES $k$ decides to transfer a task for service $s$ requested by UE $m$ to ES $k'$ for execution during frame $t$, and $\sigma_{m,s}^{k,k'}[t] = 0$ otherwise. We note that if $k' \equiv k$ (\textit{i.e.} $\sigma_{m,s}^{k,k}[t] = 1$), the task is executed locally. Finally, we introduce binary variables $\sigma_{m,s}^{k,\mathtt{c}}[t] \in \{0, 1\}$ to indicate edge-cloud cooperation, where $\sigma_{m,s}^{k,\mathtt{c}}[t] = 1$  indicates that ES $k$ offloads tasks to the cloud for execution, and $\sigma_{m,s}^{k,\mathtt{c}}[t] = 0$ otherwise.

\subsection{Service-oriented Offloading and Computational Models} \label{sec_task_offloading_model}
Let ${\Gamma _{m,s}}\left[ t_j \right] \buildrel \Delta \over = \left(T_{m}^{\max }, {a_{m,s}}\left[ t_j \right], {\omega _{m,s}}\left[ t_j \right] \right)$ denote the task associated with service $s \in \mathcal{S}$ of UE $m$ offloads to the edge layer, where $T_{m}^{\max }$ is the maximum latency requirement (seconds), $a_{m,s}\left[ t_j \right]$ is the task size (bits), and $\omega _{m,s}\left[ t_j \right]$ is the required CPU cycles (cycles). We define $\boldsymbol{\varphi} \left[ t_j \right] \triangleq \left\{ \varphi _{m,s}\left[ t_j \right] \right\}_{\forall m}$ as the portion vector of the task executed locally, where $0 \leq \varphi _{m,s}\left[ t_j \right] \leq 1, \forall m$. Consequently, $\left( 1 - \varphi _{m,s}\left[ t_j \right] \right)$  denotes the remaining portion of the task that cannot be further partitioned and is therefore offloaded to an edge server or the cloud \cite{Fan_TMC_2024}. We note that each UE $m \in \mathcal{M}$ requests only one service $s \in \mathcal{S}$ during each long-term timeslot $t$. Therefore, the requested service of UE $m$ remains unchanged during the short-term timeslots $t_j$ within the same long-term timeslot.

Given a processing rate $f_m^{\mathtt{ue}}$, the local processing time at UE $m$ can be given as
\begin{align}\label{eq:T_ue_cp}
    T_{m,s}^{\mathtt{ue,cp}}[t_j] = \frac{{\varphi_{m,s} \left[ t_j \right]{\omega _{m,s}}\left[ t_j \right]}}{{f _m^{\mathtt{ue}}}}.
\end{align}
If the offloaded task for service $s$ from UE $m$ is processed at ES $k$, the processing latency  is given as
\begin{align}\label{eq:T_k_cp}
    T_{m,s}^{k,\mathtt{cp}}[t_j] = \frac{{\left( 1 - \varphi_{m,s} \left[ t_j \right] \right){\omega _{m,s}}\left[ t_j \right]}}{{f_{mk}}}
\end{align}
where ${f_{mk}}$ (cycles/second) is the allocated processing rate from ES $k$.
For offloading to the central cloud, with a computational resource of ${f_{m\mathtt{c}}}$ (cycles/second), the corresponding processing latency is
\begin{align}\label{eq:T_c_cp}
    T_{m,s}^{\mathtt{c,cp}}[t_j] = \frac{{\left( 1 - \varphi_{m,s} \left[ t_j \right] \right){\omega _{m,s}}\left[ t_j \right]}}{{f_{m\mathtt{c}}}}.
\end{align}

To enable the servers' cooperation, each ES $k$ will process both the offloaded tasks for service $s$ from its connected users and tasks forwarded by neighboring ESs $k' \in \mathcal{K}$.
Thus, the total allocated computational resource for ES $k$ at the long-term timeslot $t$ is given as
\begin{align}\label{eq_fk}
&f_k^{\mathtt{tot}}[t] =  \sum_{m \in\mathcal{M}}\Big( 
g_s^k[t]\sigma_{m,s}^k[t] \nonumber \\
& \times  \big(1 - \sum_{k'\in\mathcal{K}\setminus \{k\}} \sigma_{m,s}^{k,k'}[t] 
- \sigma_{m,s}^{k,\mathtt{c}}[t]\big)f_{mk} \Big)\nonumber \\
& + \sum_{k'\in\mathcal{K}\setminus \{k\}} \sum_{m' \in\mathcal{M}\setminus\{m\}} \Big( 
g_s^k[t]\sigma_{m',s}^{k'}[t]\sigma_{m,s}^{k',k}[t] f_{m'k} \Big).
\end{align}  
We note that the condition $\sum_{k'\in\mathcal{K}\setminus \{k\}} \sigma _{m,s}^{k,k'}\left[ t \right] + \sigma _{m,s}^{k,\mathtt{c}}\left[ t \right] \le 1$ ensures that if ES $k$ offloads the task for service $s$ from UE $m$, it can do so to only one neighboring ES $k'\neq k$ or the cloud. Next, the total computational resource allocation of the cloud server, denoted as $f_\mathtt{c}^{\mathtt{tot}}\left[ t \right]$, is given as 
\begin{align}\label{fc_total}
\begin{split}
    f_\mathtt{c}^{\mathtt{tot}}[t] = \sum_{k \in \mathcal{K}} \sum_{m \in \mathcal{M}}
    &\sigma_{m,s}^k[t_j] \sigma_{m,s}^{k,\mathtt{c}}[t_j] f_{m\mathtt{c}}.
\end{split}
\end{align} 
Each ES $k$ and the cloud have maximum computing capacities of $F_k^{\max }$ and $F_\mathtt{c}^{\max }$, respectively.  Due to these capacity limitations, the total allocated computing resources must satisfy $f_k^{\mathtt{tot}}[t] \leq F_k^{\max }, \forall k$ and $f_\mathtt{c}^{\mathtt{tot}}[t] \leq F_\mathtt{c}^{\max }$.

\subsection{Time Delay Model}\label{sec_transmission_model}
Latency is classified into three types: Propagation delay, processing delay, and transmission delay.

\noindent\textbf{Transmission delay:} UEs and APs communicate using the frequency division multiple access (FDMA) technique. Let $W$ (Hz) denote the total system bandwidth, and  ${b_m}\left[ {{t_{j}}} \right]$ represent the bandwidth fraction  allocated to UE $m$, satisfying $\sum_{m \in\mathcal{M}} {{b_m}\left[ {{t_{j}}} \right]} \leq 1$. The channel vector between UE $m$ and AP $k$ can be modeled as $\boldsymbol{h}_m^k\left[ {{t_{j}}} \right] = \sqrt {\mathtt{PL}_m^k\left[t\right]} \,\,\bar {\boldsymbol{h}}_m^k\left[ {{t_{j}}} \right]$, where ${\mathtt{PL}_m^k\left[ {t} \right]}$ captures large-scale fading (shadowing and path loss), and $\bar {\boldsymbol{h}}_m^k\left[ {{t_{j}}} \right] \sim \mathcal{CN}\left( {0,{\textbf{I}_L}} \right)$ represents the small-scale fading following the Rayleigh fading. The signal-to-noise ratio of UE $m$ received at AP $k$ can be computed as
\begin{align}\label{eq:SNR}
 \gamma _m^k[t_{j}] = \frac{{{P_m}{{\| {\boldsymbol{h}_m^k\left[ {{t_{j}}} \right]} \|}^2}}}{{{b_m}\left[ {{t_{j}}} \right]W{N_0}}}
\end{align}
where ${P_m}$ is the UE $m$'s transmit power and ${N_0}$ denotes the noise power density. Thus, the overall uplink data rate of UE $m$ can be given as
\begin{align}\label{eq:Rm}
   {R_{m,s}}[t_{j}]=  \sum_{k \in \mathcal{K}} \sigma _{m,s}^k\left[ t \right]{b_m}\left[ {{t_j}} \right] 
 \frac{W}{{\ln 2}}\ln \left( {1 + \gamma _m^k[t_j]} \right).  
\end{align} 
Given the task size $a_{m,s}[t_j]$ and the offloading portion $({1 - {\varphi _{m,s}}[t_j]})$, the offloading transmission delay from UE $m$ to ES $k$ at time-slot $t_j$ is given as
\begin{align}\label{eq:T_mk_radio}
T_{m,s}^k\left[ {{t_j}} \right] = \frac{{\left( {1 - {\varphi _{m,s}}\left[ {{t_j}} \right]} \right){a_{m,s}}\left[ {{t_j}} \right]}}{{ R_{m,s}[t_{j}]}}.
\end{align}
Next, the transmission delay for offloading task $\sigma_{m,s}^k[t](1-\varphi _{m,s}[t_j])$ of all UEs $m\in \mathcal{M}$ from ES $k$ to the cloud server via the backhaul link is
\begin{align}\label{eq:T_kc_ts}
    T_{k,s}^{\mathtt{c}}\left[ {{t_{j}}} \right] = \frac{{\sum_{m \in \mathcal{M}} \sigma_{m,s}^k[t]\sigma_{m,s}^{k,\mathtt{c}}[t]\left( {1 - \varphi _{m,s}\left[ {{t_{j}}} \right]} \right){a_{m,s}}\left[ {{t_j}} \right]}}{{R_k^\mathtt{c}\left[ {{t_{j}}} \right]}}
\end{align} 
where ${R_k^{\mathtt{c}}\left[ {{t_{j}}} \right]}$ is the backhaul transmission rate. Similarly, given the fronthaul transmission rate ${R_k^{k'}\left[ {{t_{j}}} \right]}$,  the transmission delay for offloading task $\sigma_{m,s}^k[t](1-\varphi _{m,s}[t_j])$ of all UEs $m\in \mathcal{M}$  from ES $k$ to a neighboring server ES $k'$ is given as:
\begin{align}\label{eq:T_kk_ts}
   &T_{k,s}^{k'}[t_j] = \sum_{m \in \mathcal{M}} \bigg( \sigma_{m,s}^k[t] \sum_{k' \in \mathcal{K} \setminus \{k\}}  \sigma_{m,s}^{k,k'}[t] g_s^{k'}[t]  \nonumber \\
   &\qquad\qquad\times \left(1 - \varphi _{m,s}[t_j] \right) a_{m,s}[t_j] \bigg) / R_k^{k'}[t_j].
\end{align}
From \eqref{eq:T_mk_radio}-\eqref{eq:T_kk_ts}, the total transmission delay for the offloaded task associated with service $s$ from UE $m$ is expressed as
\begin{align}\label{T_tot_ts}
  T_{m,s}^{\mathtt{tot,t}}[t_j] = T_{m,s}^k[t_j] + \max_{\forall k} \biggr\{ T_{k,s}^{\mathtt{c}}[t_j]+ T_{k,s}^{k'}[t_j] \biggr\}.
\end{align}

\noindent\textbf{Propagation delay:}  We omit the propagation delay from UE $m$ to AP $k$ due to its short communication range. The propagation delay between ES $k$ and the cloud can be expressed as 
\begin{align}\label{eq:T_kc_pro} T_{k}^{\mathtt{c,pro}}\left[ {{t_{j}}} \right] = \frac{{{d_{k}^{\mathtt{c}}}}}{v} 
\end{align} 
where $v$ is the signal propagation speed through the fiber link over distance ${{d_{k}^{c}}}$. Similarly, the propagation delay between ES $k$ and  neighboring ES $k'$ is given as
\begin{align}\label{eqL:T_kk_pro} T_{k}^{k',\mathtt{pro}}\left[ {{t_{j}}} \right] = \frac{{{d_{k}^{k'}}}}{v} 
\end{align} 
where ${d_{k}^{k'}}$ is the corresponding distance with ${d_{k}^{k'}} \le {d_{k}^{\mathtt{c}}}$.
From \eqref{eq:T_kc_pro}-\eqref{eqL:T_kk_pro}, the overall propagation delay for the offloaded task associated with service $s$ from UE $m$ can be expressed as
\begin{align}\label{T_tot_pro}
   &T_{m,s}^{\mathtt{tot,pro}}\left[ {{t_j}} \right] = \max_{\forall k} \Big\{\sigma _{m,s}^k\left[ t \right]\sigma _{m,s}^{k,\mathtt{c}}\left[ t \right]T_{k}^{\mathtt{c,pro}}\left[ {{t_j}} \right]\nonumber \\
 &\quad+ \sigma _{m,s}^k\left[ t \right]\sum_{k' \in \mathcal{K} \setminus \{k\}} {\sigma _{m,s}^{k,k'}\left[ t \right]g_s^{k'}\left[ t \right]T_{k}^{k',\mathtt{pro}}\left[ {{t_j}} \right]}  \Big\}. 
\end{align}

\noindent\textbf{Processing delay:}
From \eqref{eq:T_ue_cp}-\eqref{eq:T_c_cp}, the total processing delay for a task of UE $m$ can be computed as
\begin{align}
T_{m,s}^{\mathtt{tot,cp}}[t_j] &= \underbrace{T_{m,s}^{\mathtt{ue,cp}}[t_j]}_{\text{Local processing time}} + \max_{\forall k} \Big\{ T_{m,s}^{k,\mathtt{cp}}[t_j] \nonumber \\
  &\times \underbrace{\sigma_{m,s}^k[t]\Big(1 - \sum_{k' \in \mathcal{K} \setminus \{k\}} \sigma_{m,s}^{k,k'}[t] - \sigma_{m,s}^{k,\mathtt{c}}[t]\Big)g_s^k[t]}_{\text{Processing at associated ES k condition}}  \nonumber \\
 &+  \underbrace{\sigma_{m,s}^k[t]\sum_{k' \in \mathcal{K} \setminus \{k\}} \sigma_{m,s}^{k,k'}[t]g_s^{k'}[t]}_{\text{Processing at ES k' condition}}T_{m,s}^{k',cp}[t_j]  \nonumber \\
& + \underbrace{\sigma_{m,s}^k[t]\sigma_{m,s}^{k,\mathtt{c}}[t]}_{\text{Processing at CS condition}}T_{m,s}^{\mathtt{c,cp}}[t_j] \Big\}. \label{T_tot_cp}
\end{align}

\noindent\textbf{The overall latency:}
The overall e2e latency for the computational task associated with service $s$ of UE $m$, accounting for propagation, processing, and transmission latencies, is expressed as:
\begin{align}
    T_{m,s}^{\mathtt{e2e}}\left[ {{t_{j}}} \right] = \underbrace{T_{m,s}^{\mathtt{tot,pro}}\left[ {{t_{j}}} \right]}_{\text{Propagation delay}} + \underbrace{T_{m,s}^{\mathtt{tot,cp}}\left[ {{t_{j}}} \right]}_{\text{Processing delay}} + \underbrace{T_{m,s}^{\mathtt{tot,t}}\left[ {{t_{j}}} \right].}_{\text{Transmission delay}}
\end{align}

\subsection{Energy Consumption Model}\label{sec_energy_consumption_model}
We also consider the total energy consumption of UE $m$, which includes local processing energy $(E_{m,s}^{\mathtt{cp}}[t_j])$ and wireless transmission energy $(E_{m,s}^{\mathtt{t}}[t_j])$. This can be expressed as \cite{ling2021}:
\begin{align}
E_{m,s}^{\mathtt{tot}}[t_j] &= E_{m,s}^{\mathtt{cp}}[t_j] + E_{m,s}^{\mathtt{t}}[t_j]  \nonumber \\
 & = \frac{{{\mu _m}}}{2}{\varphi _{m,s}}\left[ {{t_j}} \right]{\omega _{m,s}}\left[ {{t_j}} \right]{\left( {f_m^{\mathtt{ue}}} \right)^2} \nonumber \\
 &\quad + {P_m}\frac{{\left( {1 - {\varphi _{m,s}}\left[ {{t_j}} \right]} \right){a_{m,s}}\left[ {{t_j}} \right]}}{{{R_{m,s}} \left[ {{t_j}} \right] }}
\end{align}
where ${\mu _m} / 2$ is the computational energy parameter of UE $m$ (in $\text{Watt}\cdot\text{Second}^3/\text{Cycle}^3$).

\section{Problem Formulation and Analysis}\label{sec_Problem}
\subsection{The Objective Function}
 Let ${\tau ^\mathtt{u}}$ and ${\tau ^\mathtt{i}}$ denote the costs for uninstalling and installing a service on an ES ($\$/\text{service}$) \cite{Jiaming_2024}, respectively, with ${\tau ^\mathtt{u}} \leq {\tau ^\mathtt{i}}$. The change in the status of the $s$-th service on ES $k$ between the long-term timeslots $(t-1)$ and $t$ is given by $\lambda_s^k[t] = g_s^k[t] - g_s^k[t-1]$, with
\begin{align}
    \lambda_s^k\left[ t \right] = \left\{ \begin{array}{ll}
    -1, & \text{if $s$ is uninstalled;}\\
    0, & \text{if $s$ remains installed;}\\
    1, & \text{if $s$ is installed.}
    \end{array} \right.
\end{align}
As a result, we model the monetary cost of status change as
\begin{align}\label{eq:deploy_cost}
x_s^k\left[ t \right]
 = \frac{1}{2} {\left( \lambda_s^k\left[ t \right] \right)}^2 \left( {\tau ^\mathtt{u}} + {\tau ^\mathtt{i}} \right) - \frac{1}{2}\lambda_s^k\left[ t \right]\left( {\tau ^\mathtt{u}} - {\tau ^\mathtt{i}} \right).
\end{align}

Next, let ${\tau^\mathtt{o}}$ represent the cost of operating a service ($\$/\text{service}$). The cost for the $s$-th service on ES $k$ is
\begin{align}\label{operation_cost}
    x_\mathtt{o}^k\left[ t \right] = g_s^k\left[ t \right] {\tau ^\mathtt{o}}.
\end{align}

Furthermore, offloading tasks to the cloud incurs costs for data transmission and processing \cite{wang2021eihdp_price,Li_twotime_cost,Thinh_physycal_cost}. Let ${\tau^\mathtt{r}}$ denote the price per request ($\$/\text{request}$). The total request cost is:
\begin{align}\label{eq:request_cost}
    {x_\mathtt{r}}\left[ t \right] = \sum_{k \in \mathcal{K}} {\sum_{m \in \mathcal{M}} { {\sigma _{m,s}^{k,\mathtt{c}}\left[ t \right]} } } \,{\tau ^\mathtt{r}}.
\end{align}

\noindent \textbf{The total monetary cost:} Combining  \eqref{eq:deploy_cost}-\eqref{eq:request_cost}, the system's total monetary cost at the long-term time-slot $t$ is
\begin{align}
  {x^{\mathtt{tot}}}\left[ t \right] =& \sum_{k \in \mathcal{K}} {\sum_{s \in \mathcal{S}} {\left( {x_s^k\left[ t \right] + g_s^k{\tau ^\mathtt{o}}} \right)} } \nonumber\\
& + \sum_{k \in \mathcal{K}} {\sum_{m \in \mathcal{M}} { {\sigma _{m,s}^{k,\mathtt{c}}\left[ t \right]} } } {\mkern 1mu} {\tau ^\mathtt{r}},
\end{align}
satisfying $x^{\mathtt{tot}}[ t] \leq {X^{\max }}$, where ${X^{\max }}$ denotes
the maximum allowable system cost at the long-term timeslot $t$.
We define the average monetary cost over a long-term horizon $\mathcal{T} \triangleq \{1, 2, \cdots, T\}$ as
\begin{align}
    \overline x  = \mathop {\lim }\limits_{T \to  + \infty } \frac{1}{T}\sum_{t \in \mathcal{T}} {{x^{\mathtt{tot}}}\left[ {{t}} \right]}. 
\end{align}

The objective function is to minimize both the total e2e latency and monetary cost, which can be expressed as
\begin{align}\small
\eta \left( {{\boldsymbol{\delta }}\left[ t \right],{\boldsymbol{\varphi }}\left[ {{t_j}} \right],{{\boldsymbol{b}}}\left[ {{t_j}} \right]} \right) 
 ={\omega ^t}\sum_{m \in \mathcal{M}} { {T_{m,s}^{\mathtt{e2e}}\left[ {{t_j}} \right]} }
  + {\omega ^\mathtt{c}}{x^{\mathtt{tot}}}[t]
\end{align}where $\boldsymbol{\delta}\left[ {{t}} \right] \triangleq \big\{ {\sigma _{m,s}^k\left[ {{t}} \right],g_s^k\left[ t \right],\sigma _{m,s}^{k,k'}\left[ {{t}} \right],\sigma _{m,s}^{k,\mathtt{c}}\left[ {{t}} \right]} \big\}_{\forall m,k   }$ denotes the column vector encompassing all binary variables, $\boldsymbol{b}[t_j]\triangleq\{b_m[t_j]\}_{\forall m   }$, and ${\omega ^t}$ and ${\omega ^c}$ are the weight factors with ${\omega ^t} + {\omega ^\mathtt{c}} = 1$. The weight factors are determined according to the system design priorities. For example, assigning a larger $w^t$ places greater emphasis on minimizing the total e2e latency, whereas assigning a larger $w^c$ prioritizes reducing the system cost. In addition, sensitivity analysis is conducted to ensure that the selected weights achieve an appropriate trade-off between latency and cost while satisfying the latency requirements.

\subsection{Problem Formulation}
This paper tackles the challenge of optimizing joint service placement, edge/cloud cooperation, task offloading, and bandwidth allocation. The goal is to minimize both the total e2e latency and the monetary cost while ensuring compliance with specified QoS, energy consumption, and resource constraints. The problem is mathematically formulated at a short-term timeslot $t_j$ as follows:
\begin{subequations} \label{eq:Main}
	\begin{IEEEeqnarray}{cl}
		\mathop {\min }\limits_{{\boldsymbol{\delta}},\boldsymbol{\varphi} ,\boldsymbol{b}} &\quad \eta \left( {{\boldsymbol{\delta}}[t],\boldsymbol{\varphi} \left[ {{t_{j}}} \right],\boldsymbol{b}\left[ {{t_{j}}} \right]} \right)\label{eq:Maina} \\
		 \st &  T_{m,s}^{\mathtt{e2e}}\left[ {{t_{j}}} \right] \le T_m^{\max },\, \forall m                     \label{eq:Mainb}\\
		&E_{m,s}^{\mathtt{tot}}[t_j] \le E_m^{\max },\, \forall m \label{eq:Mainc} \\
        &R_{m,s}[t_j] \ge {R_{\min }},\,\forall m\label{eq:Maink} \\
        & x^{\mathtt{tot}}[ t]  \le {X^{\max }}\label{eq:Mainl} \\
        &f_k^{\mathtt{tot}}\left[ {{t}} \right] \le F_k^{\max },\,\forall k\label{eq:Mainm} \\
        &f_\mathtt{c}^{\mathtt{tot}}\left[ {{t}} \right] \le F_\mathtt{c}^{\max }\label{eq:Mainn} \\
        & {\varphi _{m,s}}\left[ {{t_j}} \right] + \sum\limits_{k \in \mathcal{K}} {\sigma _{m,s}^k\left[ t \right]} \left( {1 - {\varphi _{m,s}}\left[ {{t_j}} \right]} \right) = 1,        \forall m\label{eq:Maini} \qquad\,\\
        &{\boldsymbol{\delta }}\left[ t \right] \in \boldsymbol{\Sigma}[t], {\boldsymbol{\varphi }}\left[ {{t_j}} \right] \in \boldsymbol{\Psi} \left[ {{t_j}} \right],  \boldsymbol{b}\left[ {{t_j}} \right] \in {\mathscr {B}}\left[ {{t_j}} \right]\label{eq:Maino} 
	\end{IEEEeqnarray}
\end{subequations}
where the constraint sets $\boldsymbol{\Sigma}[t]$, $\boldsymbol{\Psi} \left[ {{t_j}} \right]$ and ${\mathscr {B}}\left[ {{t_j}} \right]$ are defined as follows:
\begin{subequations} \label{setSigma}
	\begin{IEEEeqnarray}{cl}
\boldsymbol{\Sigma}[t]\triangleq&\Big\{{\boldsymbol{\delta }}\left[ t \right] \Bigl|\, {\boldsymbol{\delta }}\left[ t \right] \in  \{0,1\},\,\label{setSigma_a}\\
&\, \sum_{k \in \mathcal{K}} {\sigma _{m,s}^k\left[ {{t}} \right] \le 1},\,\forall m, \label{setSigma_b}\\
        &\,\sum_{k' \in \mathcal{K} \setminus \{k\}} {\sigma _{m,s}^{k,k'}\left[ t \right] + \sigma _{m,s}^{k,\mathtt{c}}\left[ t \right]}  \le 1,\forall k,m, \label{setSigma_c}\quad\\
        &\,\sum_{k' \in \mathcal{K} \setminus \{k\}} {\sigma _{m,s}^{k,k'}\left[ t \right] \le 1} ,\,\forall m,k,\label{setSigma_d}\\
        &\, 1 \le \sum_{s \in \mathcal{S}} {g_s^k} \left[ t \right] \le S_k^{\max},\,\forall k \Big\}, \label{setSigma_e}
\end{IEEEeqnarray}\vspace{-10pt}
\end{subequations}        
\begin{align}        
\boldsymbol{\Psi} [t_j] \triangleq&\Big\{\boldsymbol{\varphi} \left[ {{t_{j}}} \right]\Bigl|\,  0 \le  \varphi _{m,s}[t_j] \le 1,\,\forall m\Big\}, \label{setPsi}\\
{\mathscr{B}}\left[ {{t_j}} \right] \triangleq&\Big\{\boldsymbol{b}\left[ {{t_{j}}} \right] \Bigl|\, \sum_{m \in \mathcal{M}} {{b_m}\left[ {{t_{j}}} \right]}  \le 1\, \Big\}.\label{setB}
\end{align}
In problem \eqref{eq:Main}, constraints \eqref{eq:Mainb} and \eqref{eq:Mainc} define the maximum end-to-end latency and the user's energy consumption requirements, respectively. Constraint \eqref{eq:Maink} specifies the minimum data rate for UE $m$, while constraint \eqref{eq:Mainl} ensures that the total monetary system cost remains within the allowable threshold. Constraints \eqref{eq:Mainm} and \eqref{eq:Mainn} regulate the system's computation budget. Finally, the constraint set $\boldsymbol{\Sigma}[t]$ in \eqref{setSigma} governs user association, server cooperation, and service placement decisions.

\textit{Challenges of solving problem \eqref{eq:Main}:}
The objective \eqref{eq:Maina} is inherently nonconvex, and constraints \eqref{eq:Mainb}-\eqref{eq:Maini} are also nonconvex. Consequently, problem \eqref{eq:Main} can be classified as a nonconvex MINLP, where even identifying a feasible point poses a significant challenge. Moreover, the strong coupling between binary variables ($\boldsymbol{\delta}[t]$) and continuous variables ($\boldsymbol{\varphi}[t_j], \boldsymbol{b}[t_j]$) complicates the direct application of SCA method.

\subsection{Equivalent Transformation of Problem \eqref{eq:Main}}\label{sec_Transformation}
To develop an effective solution for problem \eqref{eq:Main}, we first reformulate it into an equivalent but more tractable form. To achieve this, we introduce the following lemmas, which aim to alleviate the strong coupling between the optimization variables.
\begin{lemma}[Service Availability]\label{LM_lemma1} If service $s$ is not installed on ES $k$, the associated tasks must be offloaded to a neighboring ES $k'$ or the cloud to ensure proper execution while satisfying resource, latency, and data rate constraints. This condition is mathematically expressed as:
\begin{align}\label{lemma1_eq1}
    \sigma _{m,s}^k[t] - \Big(\sum_{k' \in \mathcal{K} \setminus \{k\}} {\sigma _{m,s}^{k,k'}[t] + \sigma _{m,s}^{k,\mathtt{c}}[t]}\Big) \le g_s^k[t].
\end{align}
\end{lemma}
\begin{proof}
The proof follows directly from the observation that if $g_s^k[t]=0$, service $s$ is not available on ES $k$ and, therefore, cannot be executed locally. When $\sigma _{m,s}^k[t]=1$, indicating that the task is offloaded to ES $k$ and it must be transferred to a neighboring ES $k'$ or the cloud, satisfying: $\sigma _{m,s}^k[t]   \le \sum_{k' \in \mathcal{K} \setminus \{k\}} {\sigma _{m,s}^{k,k'}[t] + \sigma _{m,s}^{k,\mathtt{c}}[t]} +  g_s^k[t].$
\end{proof}

\begin{lemma}[Connectivity-Dependent Offloading]\label{LM_lemma2} 
If there is no radio link connection between UE $m$ and ES $k$, tasks associated with service $s$ from this user must be executed locally or remain unprocessed. Mathematically, this condition can be expressed as:
\begin{subequations} \label{lemma2}
\begin{IEEEeqnarray}{rcl}
  \sum_{k' \in \mathcal{K} \setminus \{k\}} {\sigma _{m,s}^{k,k'}\left[ t \right] + \sigma _{m,s}^{k,\mathtt{c}}\left[ t \right]} &\le \sigma _{m,s}^k\left[ t \right],\, \forall m,k\label{lemma2_eq1}\\
    \sum_{k' \in \mathcal{K} \setminus \{k\}} {\sigma _{m,s}^{k,k'}\left[ t \right]\,} 
    &     \le \sigma _{m,s}^k\left[ t \right],\,\forall m,k\label{lemma2_eq3}\\
    \sigma _{m,s}^{k,\mathtt{c}}\left[ t \right] &\le \sigma _{m,s}^k\left[ t \right],\,\forall m,k.\label{lemma2_eq4}
\end{IEEEeqnarray}
\end{subequations}
\end{lemma}
\begin{proof}
The proof is straightforward because if $\sigma _{m,s}^k\left[ t \right] = 0$, then: $\sum_{k' \in \mathcal{K} \setminus \{k\}} {\sigma _{m,s}^{k,k'}\left[ t \right] + \sigma _{m,s}^{k,\mathtt{c}}\left[ t \right]} = 0$, $\sum_{k' \in \mathcal{K} \setminus \{k\}} {\sigma _{m,s}^{k,k'}\left[ t \right]}  = 0$ and  $\sigma _{m,s}^{k,\mathtt{c}}\left[ t \right]=0$.
\end{proof}

\begin{lemma}[Task Transfer Condition]\label{LM_lemma3}  ES $k$ can only transfer a task for service $s$ to a neighboring ES $k'$ if the latter has already installed the service, as shown by
\begin{subequations} \label{lemma3}
\begin{IEEEeqnarray}{rcl}
    \sigma _{m,s}^{k',k}\left[ t \right] &\le g_s^k\left[ t \right], \forall m, k \label{lemma3.1}\\
    \sigma _{m,s}^{k,k'}\left[ t \right] &\le g_s^{k'}\left[ t \right], \forall m, k. \label{lemma3.2}
\end{IEEEeqnarray}
\end{subequations}
\end{lemma}
\begin{proof}The proof of Lemma \ref{LM_lemma3} follows directly from Lemma \ref{LM_lemma1}.\end{proof}

\begin{proposition}\label{pro1}
From Lemmas \ref{LM_lemma1}-\ref{LM_lemma3}, the complex functions in \eqref{eq_fk}, \eqref{fc_total}, \eqref{eq:T_kc_ts}, \eqref{eq:T_kk_ts}, \eqref{T_tot_pro} and \eqref{T_tot_cp} can be simplified  as: 
\begin{subequations} \label{eq_pro1}
\begin{IEEEeqnarray}{rcl}
&&  f_k^{\mathtt{tot}}[t] =\sum_{m \in\mathcal{M}} \Big(\big(\sigma_{m,s}^k[t] -\sum_{k'\in\mathcal{K}\setminus \{k\}} \sigma_{m,s}^{k,k'}[t]   \nonumber \\
&& \qquad\quad
- \sigma_{m,s}^{k,\mathtt{c}}[t]\big)f_{mk}\Big)  + \sum_{k'\neq k} \sum_{m'\neq m} \big(\sigma_{m,s}^{k',k}[t] f_{m'k} \big)\label{reform_fk} \\
&&f_\mathtt{c}^{\mathtt{tot}}\left[ {{t}} \right] =  \sum_{k \in \mathcal{K}} {\sum_{m \in \mathcal{M}} { {\sigma _{m,s}^{k,\mathtt{c}}\left[ t \right]{f_{m\mathtt{c}}}} } }\label{fc_total_reform}\\  
&&T_{k,s}^{\mathtt{c}}\left[ {{t_{j}}} \right] = \frac{{\sum_{m \in \mathcal{M}} \sigma_{m,s}^{k,\mathtt{c}}[t]\left( {1 - \varphi _{m,s}\left[ {{t_{j}}} \right]} \right){a_{m,s}}\left[ {{t_j}} \right]}}{{R_k^\mathtt{c}\left[ {{t_{j}}} \right]}}\label{T_kc_ts_reform}\\
&&T_{k,s}^{k'}[t_j] = \frac{\sum\limits_{m \in \mathcal{M}}  \sum\limits_{k' \in \mathcal{K} \setminus \{k\}}  \sigma_{m,s}^{k,k'}[t]
 \left(1 - \varphi _{m,s}[t_j] \right) a_{m,s}[t_j] }{R_k^{k'}[t_j]}\label{T_kk_ts_reform}\quad\\
&&    T_{m,s}^{\mathtt{tot,pro}}\left[ {{t_j}} \right] =  \mathop {\max }\limits_{\forall k} \bigg\{ {\sigma _{m,s}^{k,\mathtt{c}}\left[ t \right]T_{k}^{\mathtt{c,pro}}\left[ {{t_j}} \right]} \nonumber\\
&&\qquad\qquad\quad\, +  {\sum_{k' \in \mathcal{K} \setminus \{k\}} {\sigma _{m,s}^{k,k'}\left[ t \right]T_{k}^{k',\mathtt{pro}}\left[ {{t_j}} \right]} } \bigg\}    \label{T_total__pro_reform}\\
&&T_{m,s}^{\mathtt{tot,cp}}\left[ {{t_j}} \right] =   T_{m,s}^{\mathtt{ue,cp}}\left[ {{t_j}} \right] + \max_{\forall k} \biggr\{ T_{m,s}^{k,\mathtt{cp}}\left[ {{t_j}} \right] \big( \sigma _{m,s}^k[t]\nonumber \\
&&\quad  - \sum_{k' \in \mathcal{K} \setminus \{k\}} {\sigma _{m,s}^{k,k'}\left[ t \right]}  - \sigma _{m,s}^{k,\mathtt{c}}\left[ t \right] \big) \nonumber\\
&&\,  + \sum_{k' \in \mathcal{K} \setminus \{k\}} {\sigma _{m,s}^{k,k'}\left[ t \right]} T_{m,s}^{k',\mathtt{cp}}\left[ {{t_j}} \right] + \sigma_{m,s}^{k,\mathtt{c}}\left[ t \right]T_{m,s}^{\mathtt{c,cp}}\left[ {{t_j}} \right] \biggr\}.\qquad \label{T_tot_cp_reform}
\end{IEEEeqnarray}
\end{subequations} 
\end{proposition}  
\begin{proof}Please see Appendix \ref{app:DerivationofInequ}.\end{proof}

\section{Proposed Solutions}\label{sec_Solutions}
We now propose a two-timescale optimization approach to decouple binary and continuous variables, improving solution efficiency.

\subsection{Long-term and Short-term Subproblems}\label{sec_subProblem}
As outlined earlier, our objective is to address problem \eqref{eq:Main} across different timescales. Specifically, the binary decision variables, $\boldsymbol{\delta }$, are executed during the long-term timeslot $t$ to maintain system stability in configurations and connections. In contrast, the continuous decision variables, $({\boldsymbol{\varphi }},{\boldsymbol{b}})$, are executed during the short-term timeslot $t_j$, effectively managing network uncertainties and ensuring the required QoS standards are met.

\textit{1) Long-term subproblem (L-SP):} During the long-term timeslot $t$, the focus is on optimizing service placement and edge/cloud cooperation to minimize system costs while satisfying the total e2e latency requirements for offloaded tasks. From \eqref{eq:Main} and Lemmas \ref{LM_lemma1}-\ref{LM_lemma3}, the long-term subproblem with respect to (w.r.t.) $\boldsymbol{\delta }$ can be expressed as follows:
 \begin{subequations} \label{eq:long-term}
    \begin{IEEEeqnarray}{cl}
         \textbf{L-SP: } &\mathop {\min }\limits_{\boldsymbol{\delta }[t]}\, \eta_L \left( {{\boldsymbol{\delta }}\left[ t \right]} \right)\label{eq:long-terma} \\
        &\st\quad \eqref{eq:Mainb}\text{-}\eqref{eq:Maini}, \eqref{lemma1_eq1}, \eqref{lemma2}, \eqref{lemma3}\label{eq:long-termb}\\ 
        &\qquad\ \boldsymbol{\delta}\left[ t \right] \in \boldsymbol{\Sigma}'[t]\label{eq:long-termd} 
    \end{IEEEeqnarray}
\end{subequations}
where $\boldsymbol{\Sigma}'[t]\triangleq\big\{{\boldsymbol{\delta }}\left[ t \right] \Bigl|\,\eqref{setSigma_a}, \eqref{setSigma_b}, \eqref{setSigma_d}, \eqref{setSigma_e}\big\}$. We note that $\eqref{setSigma_c}$ is  equivalently replaced by \eqref{lemma2}.

\textit{2) Short-term subproblem (S-SP):} The optimal solutions $\boldsymbol{\delta}^*$ found in \eqref{eq:long-term} is then used to solve the following short-term subproblem: 
\begin{subequations} \label{eq:short-term}
	\begin{IEEEeqnarray}{cl}
		\textbf{S-SP:}\,&\mathop {\min }\limits_{{\boldsymbol{\varphi }}[t_j],{\boldsymbol{b}}[t_j]}\ \eta_S\left( {{\boldsymbol{\varphi }}\left[ {{t_j}} \right], {{\boldsymbol{b}}}\left[ {{t_j}} \right]} \right)\label{eq:short-terma} \\
		& \st\quad \eqref{eq:Mainb}\text{-}\eqref{eq:Maink}, \eqref{eq:Maini}    \label{eq:short-termb} \\
        &\qquad\, {\boldsymbol{\varphi }}\left[ {{t_j}} \right] \in \boldsymbol{\Psi} \left[ {{t_j}} \right],  \boldsymbol{b}\left[ {{t_j}} \right] \in {\mathscr {B}}\left[ {{t_j}} \right]\label{eq:short-termc}.
	\end{IEEEeqnarray}
\end{subequations}

\begin{algorithm}[t]
\small
    \setstretch{0.95}  
    \begin{algorithmic}[1]
        \protect\caption{Overall Algorithm for Solving Problem \eqref{eq:Main}} \label{alg_1}
        \global\long\def\algorithmicrequire{\textbf{Initialization:}}    
        \REQUIRE
         Set $t = 0, j = 0$ and generate initial points (${\boldsymbol{\delta }}[0], {\boldsymbol{\varphi }}[{0}],{\boldsymbol{b}}[{0}])$ for problem \eqref{eq:Main};
        \STATE Set \textit{long-term flag} = $\mathtt{TRUE}$;
        \global\long\def\algorithmicrequire{\textbf{Main Loop:}}
        \REQUIRE 
        \FOR{each frame $t = \{1,2,\cdots,T\}$}
            \IF{\textit{long-term flag} is $\mathtt{TRUE}$}
                \STATE Solve \textbf{L-SP} \eqref{eq:long-term} to find ${{\boldsymbol{\delta }}^*}[t]$;
                \STATE Set \textit{long-term flag} = $\mathtt{FALSE}$;
            \ENDIF
            \FOR{each time-slot ${t_j} = \{{t_1},{t_2},\cdots,{t_J}\}$}
                \STATE Given $\boldsymbol{\delta }$, solve \textbf{S-SP} \eqref{eq:short-term} to find (${{\boldsymbol{\varphi }}^*[t_j]},$ ${{\bf{b}}^*[t_j]}$);
                \IF{new services are requested, or latency \\  requirements are not met} 
                    \STATE Set \textit{long-term flag} = $\mathtt{TRUE}$;
                    \STATE \textbf{break}
                \ENDIF
                \STATE Set $j = j+1$;
            \ENDFOR
            \STATE Set $t = t + 1$;
        \ENDFOR
    \end{algorithmic}
\end{algorithm}
The overall service placement and resource allocation algorithm for solving problem \eqref{eq:Main} is summarized in Algorithm \ref{alg_1} while the solutions for subproblems \eqref{eq:long-term} and \eqref{eq:short-term} will be detailed next.

\begin{remark}
    In the worst-case scenario where new services frequently arrive or latency constraints are repeatedly violated due to poor wireless channel conditions and heavy workload dynamics, Algorithm~\ref{alg_1} may trigger long-term re-optimization more frequently. In such situations, task scheduling and admission control techniques can be incorporated to regulate system load and limit the number of admitted tasks, thereby mitigating excessive re-optimization triggers. This helps prevent the framework from degenerating into an event-driven single-timescale scheme. We leave the detailed investigation of this scenario for future work.
\end{remark}

\subsection{The Proposed Solution for Solving \textbf{L-SP} \eqref{eq:long-term}}\label{sec_Solutions_longterm}
Problem \eqref{eq:long-term} is inherently a mixed-integer linear programming, and while the branch-and-bound (BnB) method ensures optimality, its complexity is prohibitive. To overcome this, we propose an efficient iterative solution leveraging SCA and penalty methods.

\textbf{Penalty function}: To address the binary nature of problem \eqref{eq:long-term}, we relax $\boldsymbol{\delta}[t]$ to continuous variables, allowing $\boldsymbol{\delta}[t] \in [0, 1]$. However, this relaxation introduces uncertainty in the optimization problem, potentially yielding solutions that are infeasible for \eqref{eq:long-term}. For any continuous variable $\delta \in (0, 1)$, we observe that $p(\delta) = \delta - \delta^2 > 0$, while $p(\delta) = \delta - \delta^2 = 0$ when $\delta \in \{0, 1\}$. Two approaches can be used to enforce $\delta = 0$ or $\delta = 1$. First, the constraint $\delta - \delta^2 \leq 0$ can be introduced, which preserves binary solutions for $\delta \in [0,1]$ but may lead to an infeasible or difficult-to-solve optimization problem. Second, a penalty function can be adopted to penalize the uncertainty of binary variables, thereby encouraging the constraint $\delta = \delta^2$ to be satisfied during the optimization process. Instead of directly incorporating these constraints into the relaxed version of \eqref{eq:long-term}, we introduce the following penalty function:
\begin{align}
    \mathcal{P}(\boldsymbol{\delta} \left[ t \right]) &= \sum_{m,k}{p}\left( {\sigma _{m,s}^k\left[ t \right]} \right) +\!\!\! \sum_{m,k,k'}\! {p}\big( {\sigma _{m,s}^{k,k'}\left[ t \right]} \big) \nonumber \\
    &\quad + \sum_{m,k} {p}\left( {\sigma _{m,s}^{k,\mathtt{c}}\left[ t \right]} \right) + \sum_{k} {p}\left( {g_s^k\left[ t \right]} \right) \geq 0.\label{penalty_function}
\end{align}
As a result, the penalized optimization problem of \eqref{eq:long-term} can be expressed as:
 \begin{subequations} \label{eq:long-term_Pen}
    \begin{IEEEeqnarray}{cl}
         \textbf{L-SP: } &\mathop {\min }\limits_{\boldsymbol{\delta }[t]}\, \eta_L \left( {{\boldsymbol{\delta }}\left[ t \right]} \right) + \alpha \mathcal{P}(\boldsymbol{\delta} \left[ t \right])\label{eq:long-term_Pena} \\
        &\st\quad \eqref{eq:Mainb}\text{-}\eqref{eq:Maini}, \eqref{lemma1_eq1}, \eqref{lemma2}, \eqref{lemma3}\label{eq:long-term_Penb}\\ 
        &\qquad\ \boldsymbol{\delta}\left[ t \right] \in \boldsymbol{\Sigma}''[t]\label{eq:long-term_Penc} 
    \end{IEEEeqnarray}
\end{subequations}
where $\boldsymbol{\Sigma}''[t]\triangleq\big\{{\boldsymbol{\delta }}\left[ t \right] \Bigl|\,\boldsymbol{\delta}[t] \in [0, 1], \eqref{setSigma_b}, \eqref{setSigma_d}, \eqref{setSigma_e}\big\}$. Herein, $\alpha$ is a positive penalty parameter to ensure that $\mathcal{P}(\boldsymbol{\delta} \left[ t \right])=0$. When a sufficiently large value of $\alpha$ is selected, the difference $\delta - \delta^2$ in the penalty function $p(\delta)=\delta-\delta^2$ is minimized, which occurs only when $\delta=0$ or $\delta=1$.

Since the penalty function $\mathcal{P}(\boldsymbol{\delta}[t])$ is nonconvex over the domain $0 \leq \boldsymbol{\delta}[t] \leq 1$, we apply the SCA method \cite{AmirBeck2010} to convexify its nonconvexity. Let $p^{(i)}(\delta)$ denote the convex surrogate of $p(\delta)$ constructed at iteration $i$ over the interval $0 \leq \delta \leq 1$. Specifically, we have
\begin{align}\label{eq:penalty_appro}
 p(\delta) =   \delta - {\delta^2} \le \delta - 2\delta{\delta^{\left( i \right)}} + {\big( {{{\delta }^{\left( i \right)}}} \big)^2} \buildrel \Delta \over = p^{(i)}(\delta).
\end{align}
Note that $p(\delta) \le p^{(i)}(\delta),\, \forall \delta^{(i)}\neq \delta^{(i-1)}$, and $p(\delta^{(i)}) = p^{(i)}(\delta^{(i)}),\, \forall \delta^{(i)}= \delta^{(i-1)}$. According to~\cite{Lipp_OptimEng_2016}, the initial value $\delta^{(0)}$ is randomly selected around $0.5$ to avoid bias toward binary assignments (true or false).
As a result, the nonconvex penalty function $\mathcal{P}(\boldsymbol{\delta} \left[ t \right])$ is iteratively replaced by the following convex one:
 \begin{align}\label{eq:Penalty_SCA}
\mathcal{P}(\boldsymbol{\delta} \left[ t \right]) &\le  \sum\limits_{m,k} {p^{\left( i \right)}\left( {\sigma _{m,s}^k\left[ t \right]} \right)} + \sum\limits_{m,k,k'} p^{\left( i \right)}\big( {\sigma _{m,s}^{k,k'}\left[ t \right]} \big) \nonumber\\
 &\quad + \sum\limits_{k,m} p^{\left( i \right)}\big( {\sigma _{m,s}^{k,\mathtt{c}}\left[ t \right]} \big) + \sum_{k} p^{\left( i \right)}\left( {g_s^k\left[ t \right]} \right)  \nonumber\\
&:= \mathcal{P}^{(i)}(\boldsymbol{\delta} \left[ t \right]).
 \end{align}
Given an positive value of $\alpha$, we iteratively solve the following approximate convex program of \eqref{eq:long-term}:
\begin{subequations} \label{eq:SCA-Longterm-all}
    \begin{IEEEeqnarray}{cl}
        & \textbf{L-SP-Convex: } \mathop {\min }\limits_{\boldsymbol{\delta }[t]}\,\, \eta_L \left( {{\boldsymbol{\delta }}\left[ t \right]} \right) + \alpha \mathcal{P}^{(i)}(\boldsymbol{\delta} \left[ t \right]) \label{eq:SCA-Longterm-a} \quad\\
        &\qquad \st\quad \eqref{eq:Mainb}\text{-}\eqref{eq:Maini}, \eqref{lemma1_eq1}, \eqref{lemma2}, \eqref{lemma3},\eqref{eq:long-term_Penc}\label{eq:SCA-Longterm-b} 
    \end{IEEEeqnarray}
\end{subequations}
until convergence and $\mathcal{P}^{(i)}(\boldsymbol{\delta} \left[ t \right])\rightarrow 0$. The proposed SCA-based iterative algorithm for solving \eqref{eq:long-term} is summarized in Algorithm \ref{alg_2}. In our numerical experiments, the relaxed solutions do not always converge to exact binary values and may instead yield near-binary results (\textit{e.g.}, $0.03$ or $0.98$). To obtain valid binary decisions, Step 6 is introduced to map these near-binary values to exact binary solutions.
To ensure feasibility, the recovered binary variables must satisfy constraints \eqref{eq:long-termb} and \eqref{eq:long-termd} \cite{Xu_2023}. Therefore, Steps 7-11 are incorporated to verify and enforce that the recovered binary variables remain feasible and satisfy all constraints of the long-term subproblem \eqref{eq:long-term}.
\begin{algorithm}[t]
\small
\caption{Proposed SCA-based Iterative Algorithm for Solving L-SP \eqref{eq:long-term}}
\label{alg_2}
\begin{algorithmic}[1]
\global\long\def\algorithmicrequire{\textbf{Initialization:}}
\REQUIRE Set $i := 0$ and generate the feasible points   
$\boldsymbol{\delta}^{(0)}\left[t\right]\triangleq\big\{\sigma _{m,s}^{k,(0)}\left[ {{t}} \right], g_s^{k,(0)}\left[ t \right],\sigma _{m,s}^{k,k',(0)}\left[ {{t}} \right],$ $\sigma_{m,s}^{k,\mathtt{c},(0)}\left[ {{t}} \right] \big\}_{\forall m,k}$; Set a positive penalty parameter $\alpha$.
\REPEAT
    \STATE Solve problem \eqref{eq:SCA-Longterm-all} to obtain optimal solutions $\boldsymbol{\delta}^*[t]$;
    \STATE Update $\boldsymbol{\delta}^{(i+1)}[t] := \boldsymbol{\delta}^*[t]$;
    \STATE Set $i := i + 1$;
\UNTIL{Convergence}
\STATE Recover exact binary solutions by:
    $({\sigma}_{m,s}^{k,*}[t]$, 
    ${\sigma}_{m,s}^{k,k',*}[t]$, 
    ${\sigma}_{m,s}^{k,\mathtt{c},*}[t]$, 
    ${g}_s^{k,*}[t]) = 
    (\lfloor{\sigma}_{m,s}^{k,(i)}[t] + 0.5\rfloor$, 
    $\lfloor{\sigma}_{m,s}^{k,k',(i)}[t] + 0.5\rfloor$, 
    $\lfloor{\sigma}_{m,s}^{k,\mathtt{c},(i)}[t] + 0.5\rfloor$, 
    $\lfloor{g}_s^{k,(i)}[t] + 0.5\rfloor)$;
\IF{Constraints \eqref{eq:long-termb} and \eqref{eq:long-termd} are satisfied}
    \STATE Break;
\ELSE
    \STATE Return to Step 1;
\ENDIF   
\STATE \textbf{Output} $({\sigma}_{m,s}^{k,*}[t]$, 
    ${\sigma}_{m,s}^{k,k',*}[t]$, 
    ${\sigma}_{m,s}^{k,\mathtt{c},*}[t]$, 
    ${g}_s^{k,*}[t])_{\forall m,k}$.
\end{algorithmic}
\end{algorithm}

\textit{Choice of the penalty parameter:} An optimal solution to \eqref{eq:SCA-Longterm-all} is guaranteed if an appropriate value of $\alpha$ is chosen. Specifically, selecting a large $\alpha$ can accelerate the convergence of Algorithm \ref{alg_2}; however, it may lead to significant performance degradation. Conversely, a small $\alpha$ ensures higher solution accuracy but results in slower convergence of the iterative algorithm.

\textit{Convergence and complexity analysis:} We denote by $\mathcal{F}^{(i)}\triangleq \{\boldsymbol{\delta}^{(i)}\big|\eqref{eq:Mainb}\text{-}\eqref{eq:Maini},\eqref{lemma1_eq1}, \eqref{lemma2}, \eqref{lemma3}, \eqref{eq:long-term_Penc}$ are feasible\} and $\mathcal{L}^{(i)} \triangleq  \eta_L ({{\boldsymbol{\delta }^{(i)}}\left[ t \right]}) + \alpha \mathcal{P}^{(i)}(\boldsymbol{\delta}^{(i)} \left[ t \right])$ the convex feasible set and the objective value of \eqref{eq:SCA-Longterm-all} found at iteration $i$, respectively. Firstly, it is observed that the approximation in \eqref{eq:Penalty_SCA} satisfies the properties of \cite{AmirBeck2010}, such as $\mathcal{P}(\boldsymbol{\delta} \left[ t \right]) \leq \mathcal{P}^{(i)}(\boldsymbol{\delta} \left[ t \right])$, $\mathcal{P}^{(i+1)}(\boldsymbol{\delta} \left[ t \right]) \leq \mathcal{P}^{(i)}(\boldsymbol{\delta} \left[ t \right])$ and $\mathcal{P}(\boldsymbol{\delta}^{(i)} \left[ t \right]) = \mathcal{P}^{(i)}(\boldsymbol{\delta}^{(i)} \left[ t \right])$. In other words, we can say that the solution obtained by Algorithm \ref{alg_2} at iteration $i$ is also feasible for problem \eqref{eq:SCA-Longterm-all} at iteration $(i+1)$, \textit{i.e.} $\mathcal{F}^{(i)} \subseteq\mathcal{F}^{(i+1)}$. Following the same argument in \cite{AmirBeck2010}, Algorithm \ref{alg_2} generates a sequence of non-increasing objective values, and $\mathcal{L}^{(i)} \rightarrow  \eta_L ({{\boldsymbol{\delta }^{(i)}}\left[ t \right]})$ when $i\rightarrow\infty$. Given that the feasible set $\mathcal{F}^{(i)}$ is compact and convex and with an appropriate value of $\alpha$, Algorithm \ref{alg_2} produces a sequence of better points that converges to at least to a local optima, satisfying the Karush-Kuhn-Tucker (KKT) conditions \cite{AmirBeck2010}. Moreover, problem \eqref{eq:SCA-Longterm-all} involves only linear constraints, thus resulting in a low computational complexity. By the interior-point method \cite[Chapter 6]{ben2001lectures}, the worst-case complexity analysis of Algorithm \ref{alg_2} per iteration is $\mathcal{O}\left(\sqrt{c}(v)^3\right)$ where $c =  4M + 3K + 4MK + 2 $ and $v = K\left( {2M + KM + 1} \right)$ are the number of linear constraints and scalar variables, respectively.

\subsection{The Proposed Solution for Solving \textbf{S-SP} \eqref{eq:short-term}} 
In problem \eqref{eq:short-term}, constraints \eqref{eq:Mainb}, \eqref{eq:Mainc}, and \eqref{eq:Maink} are non-convex. In the following, we employ SCA to convexify these constraints.

\textit{Convexity of \eqref{eq:Maink}:} We can convert $R_m[t_j]$ into the form of $z\ln( {1 + a/y})$, where $a$ is a constant.
By [\citenum{H.D.Tuan}, Eq. (73)], the function $z\ln( {1 + a/y})$ can be approximated around the feasible points $\overline z >0$ and $\overline y>0$ as follows:
\begin{align}
  z\ln \Big( {1 + \frac{a}{y}} \Big) \ge&\, 2\overline z \ln \Big( {1 + \frac{a}{{\overline y }}} \Big)  + \frac{{\overline z a}}{{a + \overline y }}\Big( {1 - \frac{y}{{\overline y }}} \Big)\nonumber \\
  &\  - \frac{{\ln \Big( {1 + \frac{a}{{\overline y }}} \Big)}}{z}{ {\overline z }^2}.\label{eq:linear1}
\end{align}
By substituting $z = {b_m}\left[ {{t_j}} \right]$, $\overline z = b_m^{(i)}\left[ {{t_j}} \right]$, $y = {b_m}\left[ {{t_j}} \right]W{N_0}$, $\overline y = b_m^{(i)}\left[ {{t_j}} \right]W{N_0}$, and $a = {P_m}{\left\| {{\bf{h}}_m^k\left[ {{t_j}} \right]} \right\|^2}$ in \eqref{eq:linear1}, the global lower bound concave function of $R_{m,s}[t_j]$ is found as
\begin{align}
  & {\small R_{m,s}[t_j] \ge\sum_{k \in \mathcal{K}}\sigma_{m,s}^k[t]\frac{W}{\ln 2} \bigg(2b_m^{(i)}[t_j]\ln\Bigl(1 + \frac{P_m\|{\bf{h}}_m^k[t_j]\|^2}{b_m^{(i)}[t_j]WN_0}\Bigr)}\,\nonumber \\
    & + \frac{b_m^{(i)}[t_j]P_m\|{\bf{h}}_m^k[t_j]\|^2}{P_m\|{\bf{h}}_m^k[t_j]\|^2 + b_m^{(i)}[t_j]WN_0} \Bigl(1 - \frac{b_m[t_j]}{b_m^{(i)}[t_j]}\Bigr) \nonumber \\
  & - \frac{\ln\Bigl(1 + \frac{P_m\|{\bf{h}}_m^k[t_j]\|^2}{b_m^{(i)}[t_j]WN_0}\Bigr)}{b_m[t_j]}(b_m^{(i)}[t_j])^2\bigg) := R_{m,s}^{(i)}[t_j] \label{eq:SCA_27k}
\end{align}
with $R_{m,s}[t_j] = R_{m,s}^{(i)}[t_j]$ whenever $b_m^{(i)}[t_j] \equiv b_m^{(i+1)}[t_j]$.
As a result, constraint \eqref{eq:Maink} iteratively replaced by the following convex constraint:
\begin{align}
    R_{m,s}^{(i)}[t_j] \ge {R_{\min }},\,\forall m.\label{Convex_27k}
\end{align}

\textit{Convexity of} \eqref{eq:Mainc}: We introduce slack variables $\boldsymbol{r}[t_j]\triangleq\{r_{m,s}[ {{t_j}}]\}_{\forall m}$, satisfying
\begin{align}
    r_{m,s}\left[ {{t_j}} \right] \ge \frac{1}{R_{m,s}^{(i)}[t_j]},\,\forall m.\label{auxilary_variable}
\end{align} 
From \eqref{auxilary_variable}, constraint \eqref{eq:Mainc} can be equivalently decomposed as 
\begin{align} \label{eq:SCA-Constraint-27c}
		 {P_m}\left( {1 - {\varphi _{m,s}}\left[ {{t_j}} \right]} \right) r_{m,s}\left[ {{t_j}} \right] {a_{m,s}[t_j]}&\nonumber \\
        +\frac{\mu _m}{2}{\varphi _{m,s}}\left[ {{t_j}} \right]{\omega _{m,s}}\left[ {{t_j}} \right]{\left( {f_m^{\mathtt{ue}}\left[ {{t_j}} \right]} \right)^2} &\le E_m^{\max }.
\end{align}
Constraint \eqref{eq:SCA-Constraint-27c} is nonconvex due to the strong coupling between $({1 - {\varphi _{m,s}}\left[ {{t_j}} \right]})$
and $r_{m,s}\left[ {{t_j}} \right]$. To convexify it, we employ the following inequality:
\begin{align}
    yz \le \frac{1}{2}\Big( {\frac{{\overline z }}{{\overline y }}{y^2} + \frac{{\overline y }}{{\overline z }}{z^2}} \Big)\label{linear2}
\end{align}
where $\bar y $ and $\bar z $ are feasible points of $y$ and $z$, respectively. Letting $y = {1 - {\varphi _{m,s}}\left[ {{t_j}} \right]} > 0,$ $\overline y  = {1 - \varphi _{m,s}^{(i)}\left[ {{t_j}} \right]} > 0, $  $z = r_{m,s}\left[ {{t_j}} \right] > 0,$ $\overline z  = r_{m,s}^{(i)}\left[ {{t_j}} \right] > 0$, constraint \eqref{eq:SCA-Constraint-27c} is iteratively replaced by
\vspace{-2mm}
\begin{align}
   &E_{m,s}^{\mathtt{tot},(i)}[t_j] \buildrel \Delta \over = P_m a_{m,s}[t_j] \frac{1}{2} \Big(\frac{r_{m,s}^{(i)}[t_j]}{(1 - \varphi_{m,s}^{(i)}[t_j])} \nonumber \\
  &\qquad\quad   \times (1 - \varphi_{m,s}[t_j])^2 + \frac{(1 - \varphi_{m,s}^{(i)}[t_j])}{r_{m,s}^{(i)}[t_j]} 
    (r_{m,s}[t_j])^2 \Big) \nonumber\quad \\
&\qquad\quad    + \frac{\mu_m}{2} \varphi_{m,s}[t_j] \omega_{m,s}[t_j] 
    (f_m^{\mathtt{ue}}[t_j])^2 \leq E_m^{\max}. \label{Convex_27c}
\end{align}

\textit{Convexity of} \eqref{eq:Mainb}: From \eqref{auxilary_variable}, we first rewrite \eqref{eq:Mainb} as
\begin{align}
    &T_{m,s}^{\mathtt{e2e}}\left[ {{t_j}} \right] \le
    T_{m,s}^{\mathtt{tot,pro}}\left[ {{t_j}} \right] + T_{m,s}^{\mathtt{tot,cp}}\left[ {{t_j}} \right] + {a_{m,s}}\left[ {{t_j}} \right] \nonumber \\
    &\times \left( {1 - {\varphi_{m,s}}\left[ {{t_j}} \right]} \right) r_{m,s}\left[ {{t_j}} \right] + \mathop {\max }\limits_{\forall k} \big\{T_{k,s}^{\mathtt{c}}\left[ {{t_j}} \right] +  {T_{k,s}^{k'}\left[ {{t_j}} \right]}  \big\}.\label{replace_r_27b}
\end{align}
and then employ \eqref{linear2} to approximate it as 
\begin{align}
    &T_{m,s}^{\mathtt{e2e}}[t_j] \le T_{m,s}^{\mathtt{tot,pro}}[t_j] + T_{m,s}^{\mathtt{tot,cp}}[t_j] + \frac{1}{2}a_{m,s}[t_j]\nonumber \\
   &\times \Big(\frac{r_{m,s}^{(i)}[t_j]}{(1 - \varphi_{m,s}^{(i)}[t_j])} (1 - \varphi_{m,s}[t_j])^2 + \frac{(1 - \varphi_{m,s}^{(i)}[t_j])}{r_{m,s}^{(i)}[t_j]}\nonumber \\
    &\times(r_{m,s}[t_j])^2\Big) + \max_{\forall k} \big\{T_{k,s}^{\mathtt{c}}[t_j]+ T_{k,s}^{k'}[t_j]\big\}  \nonumber \\
    &:= T_{m,s}^{\mathtt{e2e},(i)}[t_j]. \label{SCA_27b}
\end{align}
The function $T_{m,s}^{\mathtt{e2e},(i)}[t_j]$ is convex, resulting the following convex constraint of \eqref{eq:Mainb}
\begin{align}
    T_{m,s}^{\mathtt{e2e},(i)}\left[ {{t_j}} \right] \le T_m^{\max },\forall m.\label{Convex_27b}
\end{align}

\begin{algorithm}[t]
\small
\begin{algorithmic}[1]
\protect\caption{Proposed SCA-based Iterative
Algorithm for Solving S-SP \eqref{eq:short-term}} \label{alg_3} \global\long\def\algorithmicrequire{\textbf{Initialization:}}	
      \REQUIRE Set $i: = 0$ and generate feasible points $({\boldsymbol{\varphi }}^{(0)}[t_j],{\boldsymbol{b}}^{(0)}[t_j],\boldsymbol{r}^{(0)}[t_j])$.
	   \REPEAT
	    \STATE Solve problem \eqref{eq:SCA-short-term} to obtain the  optimal solutions $({\boldsymbol{\varphi }}^{*}[t_j],{\boldsymbol{b}}^{*}[t_j],\boldsymbol{r}^{*}[t_j])$;
	    \STATE Update $({\boldsymbol{\varphi }}^{(i)}[t_j],{\boldsymbol{b}}^{(i)}[t_j],\boldsymbol{r}^{(i)}[t_j]) := $ $({\boldsymbol{\varphi }}^{*}[t_j],{\boldsymbol{b}}^{*}[t_j],$ $\boldsymbol{r}^{*}[t_j])$.
	    
	    \STATE Set $i: = i + 1$;
	    
	    \UNTIL Convergence 
           \STATE \textbf{Output} $({\boldsymbol{\varphi }}^{*}[t_j],{\boldsymbol{b}}^{*}[t_j])$.
\end{algorithmic}
\end{algorithm}

Based on the above analysis, the convex program of \eqref{eq:short-term} solved at iteration $i$ is given as
\begin{subequations} \label{eq:SCA-short-term}
	\begin{IEEEeqnarray}{cl}
&\textbf{S-SP-Convex:}\		 \mathop {\min }\limits_{{\boldsymbol{\varphi }},{\boldsymbol{b}},{\boldsymbol{r}}} {\omega ^t}{\sum_{m \in \mathcal{M}} {{T_{m,s}^{\mathtt{e2e},(i)}\left[ {{t_j}} \right]} } }  + {\omega ^c}\bar x\label{eq:SCA-short-terma} \qquad\\
		&\qquad\quad \st\  
        \eqref{eq:Maini}, \eqref{eq:short-termc}, \eqref{Convex_27k}, \eqref{auxilary_variable}, \eqref{Convex_27c}, \eqref{Convex_27b}.\label{eq:SCA-short-termb} \qquad
	\end{IEEEeqnarray}
\end{subequations}
The proposed iterative algorithm for solving problem \eqref{eq:short-term} is summarized in Algorithm \ref{alg_3}. Its convergence can be established similarly to that of Algorithm \ref{alg_2}. Problem \eqref{eq:SCA-short-term} only involves $3M$ scalar variables and $5M + K + 3$ linear and second-order cone (SOC) constraints. As a result, the computational complexity per iteration of Algorithm \eqref{alg_3} is $\mathcal{O}\big(\sqrt{5M + K + 3}(3M)^3\big)$.

\section{Numerical Results}\label{sec_NumericalResults}
 \subsection{Network Setting and Benchmark Schemes}\label{sec_simulation_setup}
 We consider an HECC-aided IoT network where all UEs and APs are located within a small-cell scenario covering an area of $200\times 200$ m. The ES locations vary based on the number of servers: For $K = 2$, they are positioned at $(66,100)$ and $(133,100)$, while for $K = 4$, they are at $(40,100)$, $(80,100)$, $(120,100)$, $(160,100)$ for $K = 4$. The distance from ES $k$ to the central cloud is fixed at
 ${d_{k,c}} = 10$~km. The pathloss between UE $m$ and AP $k$ is modeled as $\mathtt{PL}_m^k  = - 35.3 - 37.6\log_{10}({d_{mk}})$ measured in dB, where $d_{mk}$ is the corresponding distance \cite{H.D.Tuan}.  The noise spectral density is set to -$174$ dBm/Hz.
 
 Following\cite{huynh_latency}, the maximum computing capacity of the cloud center is set to $F_c^{\max } = 30$ GHz. The propagation speed is set to $v = 2 \times 10^8$ m/s and the number of time-frames is ${T_J} = 50$. UEs randomly request services from ESs or the central cloud with a cycle of five time-frames. The service request price for data transfer from ESs to the central cloud, based on the AWS pricing for intra-region transfers within the US East (New York City) region, is $\tau^\mathtt{r} = \$0.01$   \cite{aws_ec2_pricing}. However, this pricing model can be easily adapted for other edge-cloud computing providers. Other parameters are listed in Table \ref{tab:Simulationparameter1}, following studies in \cite{huynh_latency, huynh_service_placement, Zhou_cost, cooperation2017,Zhang_2021, 3GPP, H.D.Tuan, Jiaming_2024, Liu_resource_frame}.

The proposed algorithms are implemented in the MATLAB environment. We use the MOSEK solver within the CVX modeling toolbox to solve all convex programs.
 \begin{table}[t]
\centering  
\captionof{table}{Simulation Parameters}
\label{tab:Simulationparameter1}
\scalebox{0.77}{
    \begin{tabular}{l|l}
        \hline
        Parameter & Value \\
        \hline\hline
        System bandwidth, $W$ & $10$ MHz \\
        Number of UEs, $M$ & $\{8, 10 \}$ \cite{huynh_service_placement}\\
        Number of antennas at each AP, $L$&$\{8, 10 \}$\\
        Maximum number of services, $S$  & $8$\\
        Maximum number of services in ES, $S_k^{\max}, \forall k$ & $6$\\ 
        Processing rate of UE $\&$ ES, $f_m^{\mathtt{ue}}$ $\&$ ${f_{mk}},\forall k$  & $1$ GHz $\&$ $2$ GHz \\
         Maximum computing capacity of ESs, $F_k^{\max }$ & $\{8,20\}$ GHz\\
        Cloud processing rate, ${f_{m\mathtt{c}}}$ & $4$ GHz \cite{cooperation2017,Zhang_2021}\\
        Backhaul $\&$ fronthaul capacity,
        $R_k^{k'}$ $\&$ $R_k^\mathtt{c}$, $\forall k$ & $5$ $\&$ $1$ Gbps \\
        Task size, ${a_{m,s}},\,\forall m, s$ & $1354$ bytes \cite{3GPP} \\
        Maximum delay requirement, $T_m^{\max },\,\forall m$ & $2$ ms\\
        Task complexity, ${\raise0.7ex\hbox{${{w_m}}$} \!\mathord{\left/
        {\vphantom {{{w_m}} {{a_{m,s}}}}}\right.\kern-\nulldelimiterspace}
        \!\lower0.7ex\hbox{${{a_{m,s}}}$}},\forall m,s$  & $\left[ {200,500} \right] \frac{cycles}{byte}$ \\
        UE's transmitted power, ${P_m},\,\forall m$ & $23$ dBm \\
        Minimum data rate requirement, ${R_{\min }}$ & $1$ Mbps \\
        Price for installing a service, ${\tau ^\mathtt{i}}$ & $0.1$ $\$/$service \cite{Jiaming_2024}\\
        Price for uninstalling a service, ${\tau ^\mathtt{u}}$ & $0.05$ $\$$/service \cite{Jiaming_2024} \\
        Price for operating a service, ${\tau ^\mathtt{o}}$ & $0.1$ $\$$/service \cite{Jiaming_2024} \\
        Maximum average monetary cost, ${X^{\max }}$ & $20$ ${(\$)}$ \cite{Jiaming_2024}\\
        UE's Maximum energy consumption, $E_m^{\max }$ & $1$ Joule \\
        Effective capacitance coefficient, ${\mu _m}$ & ${10^{ - 27}}\,\frac{Watt.{s^3}}{cycl{e^3}}$\\
        \hline		  				
    \end{tabular}
}
\end{table}

\textit{Benchmark Schemes:} For performance comparison, we consider the following benchmark schemes:
\begin{itemize}
    \item \textbf{Branch-and-Bound (BnB)} is an optimal algorithm for solving NP-Hard problems but with worst-case exponential complexity \cite{Dinh_2019}. 
    \item \textbf{Fixed User Association (FUAS)}: This scheme selects the strongest wireless channel between a UE and ESs \cite{huynh_latency} and maintains the connection over a long-term timeslot.
    \item \textbf{Random User Association (RUAS)} assigns UEs to ESs randomly for task offloading \cite{huynh_latency}.
    \item  \textbf{No Edge-Edge Cooperation (NEEC)} disables ES cooperation to assess its impact on performance.
    \item \textbf{Without Cloud (W/o Cloud)} processes tasks only at ESs, evaluating the cloud's role in reducing cost and e2e latency \cite{huynh_latency}.
    \item \textbf{Equal Bandwidth (EB)} allocates equal bandwidth to all UEs to study its effect on e2e latency.
    \item \textbf{Fixed Offloading Rates} ($30\%, 80\%, 100\%$) enforces specific offloading rates from UEs to ESs to analyze performance under varying workloads.
\end{itemize}

\subsection{Convergence Behavior}\label{sec_numerical_results_dicussions}
\begin{figure}[t]
    \centering
    \vspace{-2mm}
    \subfigure[Impact of $\alpha$ with $M$ $=$ $8$ UEs and $K$ $=$ $2$ ESs]{
        \includegraphics[width = 2.5in]{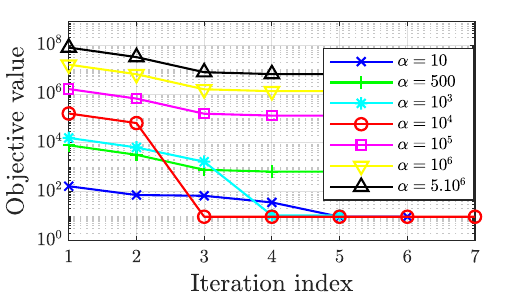}
        \label{fig:Convergence_penalty}
    }
    \subfigure[Convergence behavior with $\alpha = 10^4$ ]{
        \includegraphics[width =2.5in]{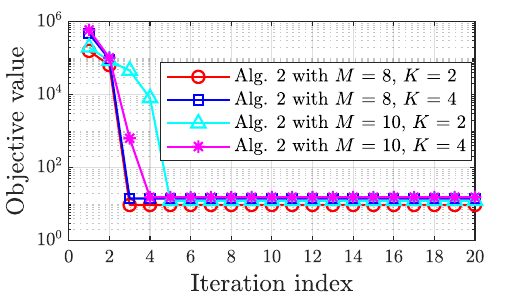}
        \label{fig:Convergence_algorithm2}
    }
      \vspace{-2mm}
    \caption{\small Convergence analysis of Algorithm \ref{alg_2}. }
    \label{fig:Convergence_algorithm2_penalty}
\end{figure}

In Fig. \ref{fig:Convergence_algorithm2_penalty}, we analyze the convergence of Algorithm \ref{alg_2} under varying penalty parameters $\alpha$, and 
 numbers of UEs and ESs. As shown in Fig.~\ref{fig:Convergence_penalty}, each value of $\alpha$ produces a decreasing sequence of objective values, converging in about five iterations. Larger $\alpha$ values (\textit{e.g.} $10^6$) yield higher objective values due to excessive penalties, which increase tolerance for binary variables ${\boldsymbol{\delta}}[t]$. Conversely, smaller $\alpha$ values (\textit{e.g.} $10, 500, 1000$) create a mismatch between the penalty function \eqref{penalty_function} and objective function \eqref{eq:long-term}, slowing convergence. The optimal convergence occurs with $\alpha = 10^4$, achieving the minimum objective value in four iterations due to reduced error tolerance. As a result, we choose this $\alpha = 10^4$ to compare convergence across different schemes. Fig. \ref{fig:Convergence_algorithm2} shows that convergence slightly slows for ($M = 10, K = 2$) and ($M = 10, K = 4$), indicating modest sensitivity to problem size.

\begin{figure}[!htb]
	\centering
    \includegraphics[width=0.8\columnwidth,trim={0cm 0.0cm 0cm 0.0cm}]{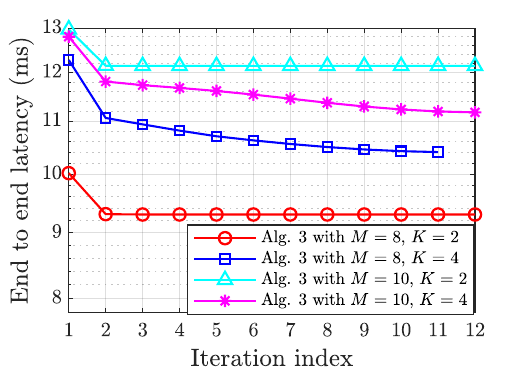}
    \vspace{-2mm}
	\caption{\small Convergence behavior of Algorithm \ref{alg_3}.}
	\label{fig:Convergence_algorithm3}
\end{figure}

We examine the impact of network size on the convergence behavior of Algorithm~\ref{alg_3}. As can be seen from Fig.~\ref{fig:Convergence_algorithm3}, the convergence rate declines more noticeably as the number of UEs and ESs increases, indicating that larger networks require more iterations to achieve convergence.

\subsection{Performance Comparison}
\begin{figure}[h]
    \centering
    \subfigure[\small The objective value versus timeslot, $t$]{
        \includegraphics[width = 2.5in]{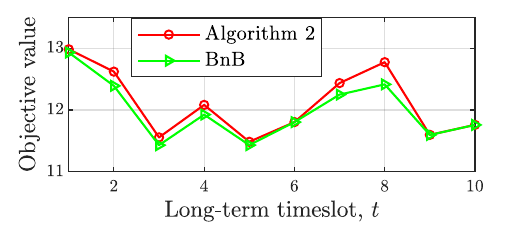}
        \label{fig:compare_Al2_BB_overall}
    }
    \subfigure[\small The cost value versus timeslot, $t$]{
        \includegraphics[width =2.5in]{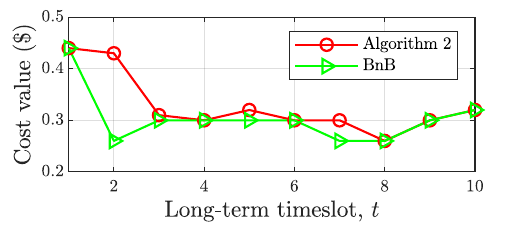}
        \label{fig:compare_Al2_BB_cost}
    }
    \vspace{-2mm}
    \caption{\small Performance comparison between Algorithm~\ref{alg_2} and BnB with $M = 8$ UEs and $K = 2$ ESs.}
    \label{fig:compare_BB}
\end{figure}

In Fig.~\ref{fig:compare_Al2_BB_overall}, we compare the performance of Algorithm \ref{alg_2} with BnB. BnB achieves optimal service placement, user association, and edge-cloud cooperation, minimizing cost and end-to-end latency over ten long-term timeslots. However, Algorithm~\ref{alg_2} deviates by only $1$-$2\%$, demonstrating highly effective binary decisions with significantly lower complexity. Fig.~\ref{fig:compare_Al2_BB_cost} further compares system costs, showing that both algorithms exhibit similar trends in response to user requests and the uncertainty of wireless channels. This confirms that our approach provides near-optimal solutions while maintaining high system performance.

\begin{figure}[t]
		\subfigure[System cost versus $t_j$]
		{
			\includegraphics[scale=0.46]{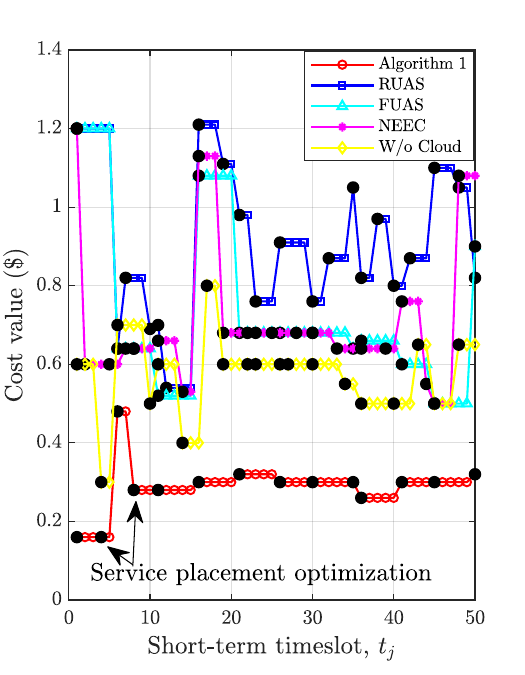}
            \label{fig:overall_cost}
		}\hspace{-15pt}
		\subfigure[Total users' e2e latency versus, $t_j$]
		{
			\includegraphics[scale=0.46]{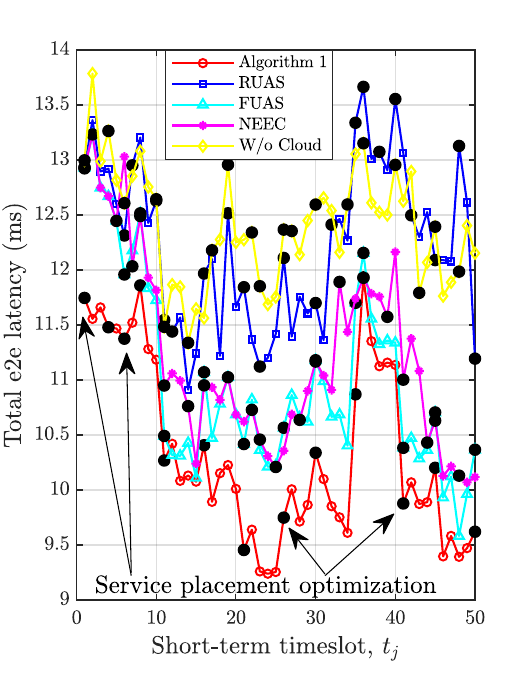}
			\label{fig:overall_latency_time_index}
		}
		\caption{Impacts of the user association and edge/cloud cooperation on service placement with $M = 8$ UEs and $K = 2$ ESs. }
		\label{fig:impact_binary_Schemes}
\end{figure} 
Fig. \ref{fig:impact_binary_Schemes} examines the impact of user association and edge/cloud cooperation on service placement. As shown in Fig. \ref{fig:overall_cost}, Algorithm~\ref{alg_1} achieves the lowest and most stable system cost across timeslots, even under dynamic network conditions. In contrast, RUAS has the highest cost with greater fluctuations due to random user association, frequently triggering service placement optimization to meet e2e latency and offloading requirements.
FUAS, which selects the best wireless channel, offers more stability than RUAS but increases costs when ESs require additional software or cloud transfers. Similarly, NEEC and W/o Cloud see cost surges of $200$-$250\%$ compared to Algorithm~\ref{alg_1} due to the lack of edge-edge or edge-cloud cooperation.
Fig. \ref{fig:overall_latency_time_index} shows frequent e2e latency fluctuations across all schemes due to dynamic network conditions. Algorithm~\ref{alg_1} significantly outperforms the benchmark schemes in terms of latency reduction while requiring the fewest long-term re-optimization triggers. In contrast, the W/o Cloud and RUAS schemes experience the highest delays due to the lack of network optimization.

\begin{figure}[t]
    \centering
    \subfigure[\small CDF versus $T^{\max}\equiv T_m^{\max },\, \forall m $]{
        \includegraphics[width = 2.5in]{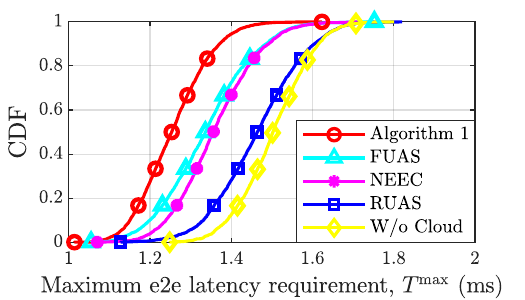}
        \label{fig:overall_latency}
    }
    \subfigure[\small CDF versus $X^{\max}$]{
        \includegraphics[width =2.5in]{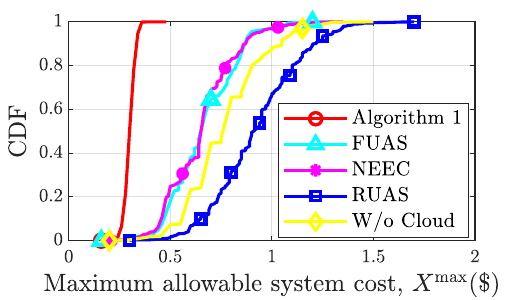}
        \label{fig:CDF_Al2_BB_cost}
    }
    \vspace{-2mm}
    \caption{\small Cumulative distribution function with $M = 8$ UEs and $K = 2$ ESs.}
    \label{fig:overall_system cost}
\end{figure}

Fig. \ref{fig:overall_latency} presents the cumulative distribution function (CDF) of the maximum users' e2e latency threshold, $T^{\max}\equiv T_m^{\max }, \forall m $. Results show that feasibility probabilities decrease as $T^{\max}$ decreases. Algorithm \ref{alg_1} outperforms all other benchmark schemes, maintaining lower e2e latency and better fairness among UEs compared to W/o Cloud and RUAS. FUAS and NEEC achieve an approximate 0.2 ms latency offset, surpassing RUAS and W/o Cloud in $90\%$ of trials. These findings highlight the importance of jointly optimizing user association, edge/cloud cooperation, and service placement for improved QoS. Fig. \ref{fig:CDF_Al2_BB_cost} shows the CDF as a function of the system cost threshold, $X^{\max}$, assessing sensitivity to cost variations. RUAS has the lowest feasibility, while Algorithm \ref{alg_1} achieves the highest for lower $X^{\max}$ values. NEEC, FUAS, and W/o Cloud exhibit system cost $\$0.5$-$0.7$ lower than RUAS in $50\%$ of trials. These results confirm that random user association significantly impacts service placement optimization, leading to higher system costs.

\begin{figure}[t]
	\centering
        \vspace{-2mm}
	\includegraphics[width=0.8\columnwidth,trim={0cm 0.0cm 0cm 0.0cm}]{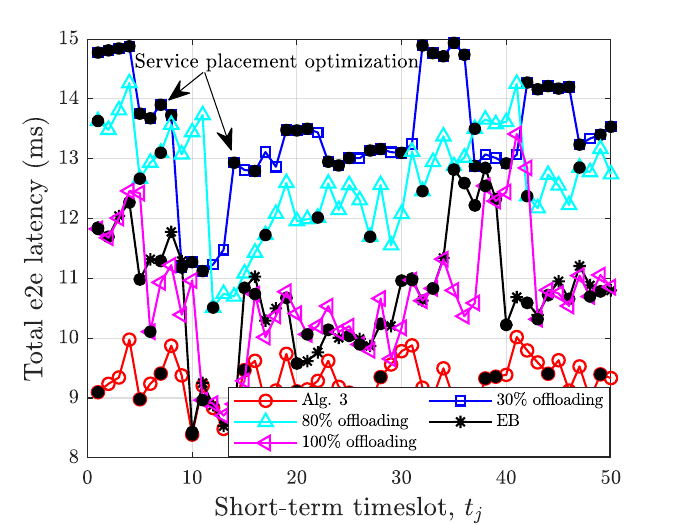}
	\caption{\small The e2e latency between Algorithm \ref{alg_3} and benchmark schemes with $M = 8$ UEs and $K = 2$ ESs.}
	\label{fig:compare_Al3}
\end{figure}
\textit{Impact of task offloading:} Fig. \ref{fig:compare_Al3} examines the impact of task offloading on the evaluated schemes. As shown, the total e2e latency increases by approximately 50\%, 40\%, and 25\% compared to Algorithm \ref{alg_3} when task offloading is fixed at 30\%, 80\%, and 100\%, respectively. This suggests that increasing task offloading to ESs under unpredictable network conditions significantly reduces latency while maintaining stable ES configurations. However, the 100\% offloading scheme does not achieve the lowest latency due to its fixed proportion across all UEs, leading to a suboptimal solution. Additionally, the EB scheme experiences a $5$-$10\%$ latency increase compared to Algorithm~\ref{alg_3} due to equal bandwidth allocation instead of optimization. Clearly, Algorithm~\ref{alg_3}, which jointly optimizes offloading and bandwidth allocation, achieves the lowest total e2e latency among all schemes.

\begin{figure}[t]
	\centering
	\vspace{-2mm}\includegraphics[width=0.8\columnwidth,trim={0cm 0.0cm 0cm 0.0cm}]{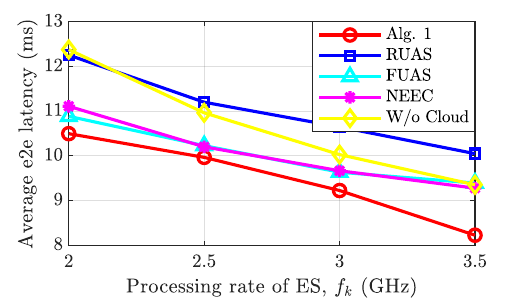}
	\caption{\small The average e2e latency versus the processing rate $f_k$, with $M = 8$ UEs and $K = 2$ ESs.}
	\label{fig:compare_latency_behavior_different_F_long}
\end{figure}
\begin{figure}[t]
	\centering
	\vspace{-2mm}\includegraphics[width=0.8\columnwidth,trim={0cm 0.0cm 0cm 0.0cm}]{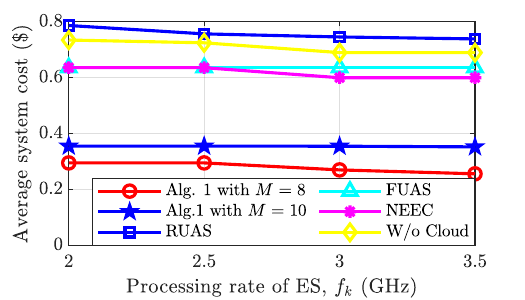}
	\caption{\small The average system cost of different service placement and user association schemes versus the processing rate $f_k$ of ESs, with $M = \{8, 10\}$ UEs and $K = 2$ ESs.}
	\label{fig:compare_cost_behavior_different_F_long}
\end{figure} 

\textit{Impact of processing rate:} Fig.~\ref{fig:compare_latency_behavior_different_F_long} shows that the average total e2e latency across all schemes is reduced significantly as the ESs' processing rate $f_{mk}\equiv f_k, \forall k $ increases from $2$ GHz to $3.5$ GHz, primarily due to reduced task processing times at higher ES speeds. Similarly, Fig. \ref{fig:compare_cost_behavior_different_F_long} analyzes the impact of ES processing rates on system costs. With 
$M = 8$ UEs and $K = 2$ ESs, system costs gradually decrease as higher processing rates allow more tasks to be executed locally, reducing cloud offloading and service installation costs. However, in scenarios like FUAS, Algorithm \ref{alg_1} maintains constant costs. This is because FUAS assigns UEs to the nearest ESs, and the increased UE count keeps costs stable despite changes in ESs' processing rates.

\begin{figure}[t]
	\centering
	\includegraphics[width=1\columnwidth,trim={0cm 0.0cm 0cm 0.0cm}]{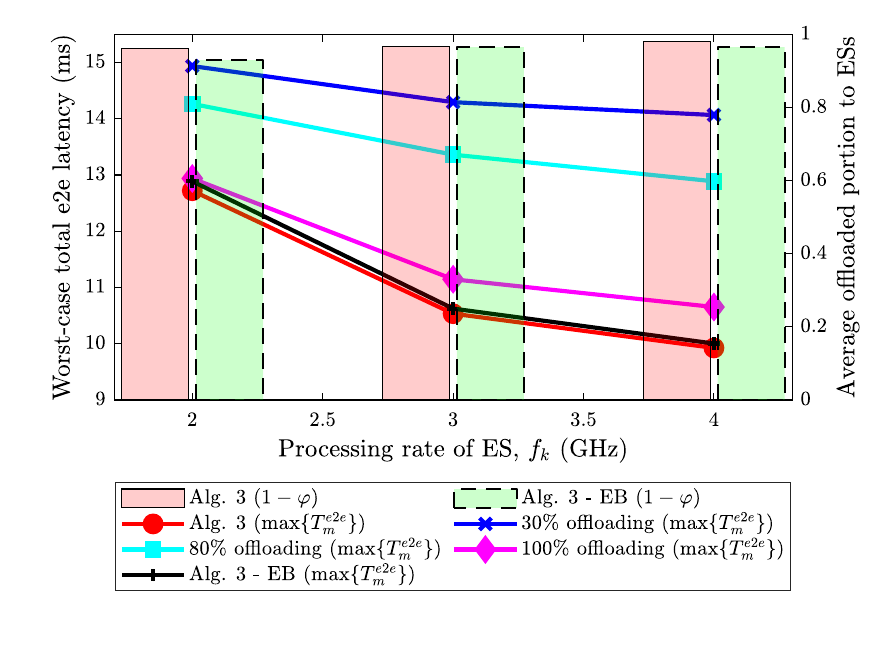}
	\caption{\small The worst-case total e2e latency and average offloading portion versus different ESs processing rate $f_k$, with $M = 8$ UEs and $K = 2$ ESs.}
    \label{fig:compare_latency_offload_behavior_different_f}
\end{figure}
We further examine the impact of ESs' processing rates on task offloading behavior and worst-case total e2e latency across different resource allocation schemes, as shown in Fig.~\ref{fig:compare_latency_offload_behavior_different_f}. The results indicate that as the ESs' computational resources increase from $2$ GHz to $4$ GHz, the average offloading proportion from UEs to ESs rises slightly, while the worst-case total e2e latency decreases significantly. Notably, Algorithm~\ref{alg_3} maintains a consistently higher average offloading proportion than the EB scheme. However, despite $100\%$ task offloading without optimization, total e2e latency remains unchanged. This underscores the importance of jointly optimizing bandwidth and task offloading allocation to minimize latency.

\begin{figure}[h]
	\centering
    \includegraphics[width=0.9\columnwidth,trim={0cm 0.0cm 0cm 0.0cm}]{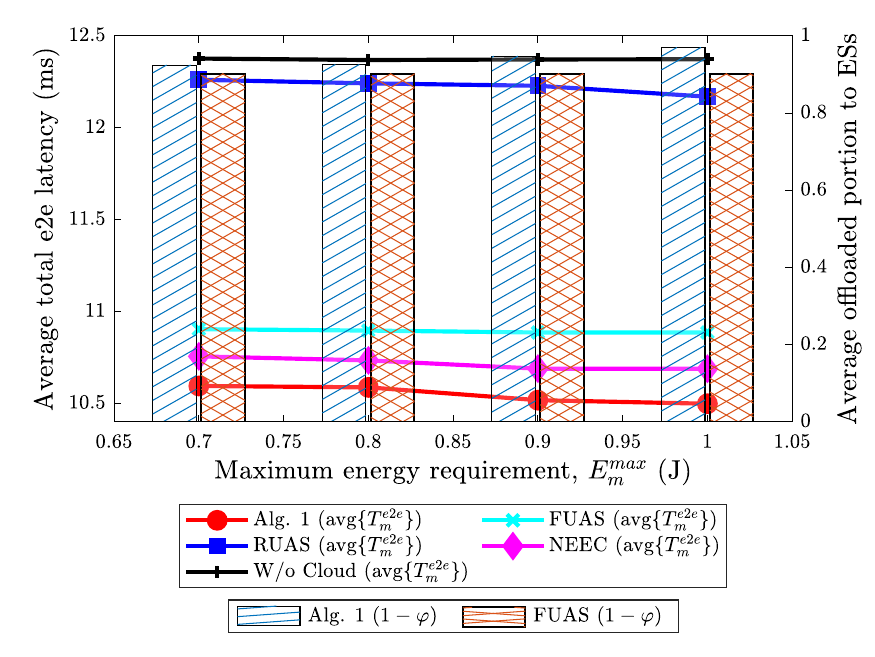}
	\caption{\small The average total e2e latency versus the maximum energy consumption requirement $E_m^{\max }$, with $M = 8$ UEs and $K = 2$ ESs.}
\label{fig:compare_latency_behavior_different_E_long}
\end{figure}

\textit{Impact of maximum energy requirement:} As shown in Fig. \ref{fig:compare_latency_behavior_different_E_long}, Algorithm \ref{alg_1} significantly reduces the average total e2e latency. This is attributed to the fact that a higher task offloading portion, $(1-\varphi \left[ {{t_j}} \right])$, can be more effectively allocated when UEs have a larger $E_m^{\max}$ limit, given that all binary decision variables are jointly optimized while satisfying constraint~\eqref{eq:Mainc}. In contrast, other schemes show only slight latency reductions or remain stable due to long-term network configuration constraints. The absence of edge-edge cooperation, fixed or random user association, and lack of edge-cloud cooperation hinder short-term optimization, restricting improvements in task offloading efficiency.

\section{Conclusion}\label{sec_conclusion}\vspace{-2pt}
In this paper, we have explored joint service placement and resource allocation in HECC-aided IoT networks to minimize the total e2e latency and system cost. We have formulated an optimization problem integrating service placement, user association, edge/cloud cooperation, task offloading, and bandwidth allocation while accounting for ESs' computational constraints and UE energy and latency requirements. Given the problem's NP-hard nature due to binary-continuous variable coupling, we have introduced new lemmas to simplify the formulation and developed a two-timescale optimization framework. The proposed framework applies the SCA method to convexify non-convex constraints and employs penalty functions to enforce binary decisions. Numerical results demonstrate fast convergence and significant performance improvements, underscoring the effectiveness of joint service placement and resource allocation in HECC-aided IoT networks. Future work will explore deep reinforcement learning methods to solve the mixed-integer optimization problem involving service placement, user association, and edge-cloud cooperation in large-scale networks.

\appendices
\renewcommand{\thesectiondis}[2]{\Alph{section}:}
\section{Proof of Proposition \ref{pro1}} \label{app:DerivationofInequ}
\renewcommand{\theequation}{\ref{app:DerivationofInequ}.\arabic{equation}}\setcounter{equation}{0}\vspace{-2pt}

\begin{table}[htb]
\centering  
\captionof{table}{True Table}
\label{tab:true_table}
    \begin{tabular}{@{}cc|c@{}}
        \hline
        $A$ & $B$ & $Y = AB$ \\
        \hline\hline
       0 & 0 & 0 \\
       1 & 0 & 0 \\
       1 & 1 & 1 \\
        \hline                          
    \end{tabular}\vspace{-5pt}
\end{table}
First, let us consider the above truth table, where $A \in \left\{ 0,1 \right\}$, $B \in \left\{ 0,1  \right\}$, and $B \leq A$.  From Table \ref{tab:true_table}, it follows that $Y = B$, and then
\begin{align}\label{proof_logic_simplify}
    \left\{ {\begin{array}{*{20}{c}}
{Y = AB}\\
{B \le A}
\end{array}} \right. \Leftrightarrow \left\{ {\begin{array}{*{20}{c}}
{Y = B}\\
{B \le A}.
\end{array}} \right.
\end{align}
From Lemmas \ref{LM_lemma1}-\ref{LM_lemma3} and \eqref{proof_logic_simplify}, we are now in position to show derivations in \eqref{eq_pro1}.

\textit{Derivation of \eqref{reform_fk}:} Eq. \eqref{eq_fk} can be rewritten as $f_k^{\mathtt{tot}}\left[ {{t}} \right] = {Y_1} + {Y_2}$,
where 
\begin{align}
   Y_1 \triangleq & \sum_{m \in \mathcal{M}} \Big(g_s^k[t] \big(\sigma_{m,s}^{k}[t] - \sigma_{m,s}^{k}[t] \sum_{k' \in \mathcal{K} \setminus \{k\}} \sigma_{m,s}^{k,k'}[t]\nonumber\\
    &- \sigma_{m,s}^{k}[t] \sigma_{m,s}^{k,\mathtt{c}}[t]\big)  f_{mk}\Big)\label{Y1}\\ 
Y_2 \triangleq & \sum_{k' \in \mathcal{K} \setminus \{k\} } \sum_{m' \in \mathcal{M}} \Big(g_s^k[t] \sigma_{m',s}^{k'}[t] \sigma_{m',s}^{k',k}[t]  f_{m'k}\Big)\label{Y2}.
\end{align} 
From \eqref{proof_logic_simplify} and \eqref{lemma2_eq3}, we can decouple $\sigma_{m,s}^k\left[t \right]\sum_{k' \in \mathcal{K} \setminus \{k\}} {\sigma _{m,s}^{k,k'}\left[ t \right]}$ in $Y_1$ as follows:
\begin{equation}\label{derive_Y1_p1}
    \left\{ \begin{array}{l}
    X_1 = \sigma_{m,s}^k[t]  \!\!\!\sum\limits_{k' \in \mathcal{K} \setminus \{k\}} \!\!\!\sigma_{m,s}^{k,k'}[t] \\
    \eqref{lemma2_eq3}
    \end{array} \right.
    \!\!\!\Leftrightarrow\!
    \left\{ \begin{array}{l}
    X_1 = \!\!\!\sum\limits_{k' \in \mathcal{K} \setminus \{k\} } \!\!\!\sigma_{m,s}^{k,k'}[t] \\
    \eqref{lemma2_eq3}.
    \end{array} \right.
\end{equation}
By \eqref{lemma2_eq4}, we can equivalently rewrite $\sigma_{m,s}^k[t]\sigma_{m,s}^{k,\mathtt{c}}[t]$ as
\begin{equation}\label{derive_Y1_p2}
    \left\{ \begin{array}{l}
    X_2 = \sigma_{m,s}^k[t]\sigma_{m,s}^{k,\mathtt{c}}[t] \\
    \eqref{lemma2_eq4}
    \end{array} \right.
    \Leftrightarrow
    \left\{ \begin{array}{l}
    X_2 = \sigma_{m,s}^{k,\mathtt{c}}[t] \\
    \eqref{lemma2_eq4}.
    \end{array} \right.
\end{equation}
From \eqref{derive_Y1_p1} and \eqref{derive_Y1_p2}, $Y_1$ is simplified as follows: 
\begin{equation}\label{Y1_reform_1}
\begin{split}
Y_1 = &\sum_{m \in \mathcal{M}} \big(g_s^k[t] (\sigma_{m,s}^{k}[t] - 
    \sum_{k' \in \mathcal{K} \setminus\{k\}} \sigma_{m,s}^{k,k'}[t] \\
    &\qquad\, - \sigma_{m,s}^{k,\mathtt{c}}[t] )   f_{mk}\big)
\end{split}
\end{equation} 
Next, from the service availability \eqref{lemma1_eq1}, we decouple $g_s^k[t] (\sigma_{m,s}^{k}[t] - 
    \sum_{k' \in \mathcal{K} \setminus\{k\}} \sigma_{m,s}^{k,k'}[t]  - \sigma_{m,s}^{k,\mathtt{c}}[t] )$ as
\begin{equation}\label{derive_Y1_p3}
\begin{aligned}
&\quad\, \left\{
    \begin{aligned}
        & X_3 = g_s^k[t] \Big(\sigma_{m,s}^{k}[t] - 
        (\!\!\!\!\sum_{k' \in \mathcal{K} \setminus \{k\}} \sigma_{m,s}^{k,k'}[t] + \sigma_{m,s}^{k,\mathtt{c}}[t])\Big) \\
        & \eqref{lemma1_eq1}
    \end{aligned}
\right. \\
&\Leftrightarrow \left\{
    \begin{aligned}
        &X_3 = \sigma_{m,s}^{k}[t] - 
        \big(\!\!\!\!\sum_{k' \in \mathcal{K} \setminus \{k\}} \!\!\!\!\sigma_{m,s}^{k,k'}[t] + \sigma_{m,s}^{k,\mathtt{c}}[t]\big) \\
        &\,\,\, \!\!\!\!\eqref{lemma1_eq1}.
    \end{aligned}
\right.
\end{aligned}
\end{equation}
to rewrite $Y_1$ as 
\begin{align}\label{Y1_final}
Y_1 = &\sum_{m \in \mathcal{M}} \Big(\sigma_{m,s}^{k}[t] - 
    \big(\sum_{k' \in \mathcal{K} \setminus \{k\}} \sigma_{m,s}^{k,k'}[t] +\sigma_{m,s}^{k,\mathtt{c}}[t]\big) \Big)f_{mk}.
\end{align} 
Similarly, by \eqref{proof_logic_simplify} and \eqref{lemma2_eq3}, we can simplify to $Y_2$ as 
\begin{equation}\label{Y2_final}
Y_2 = \!\!\!\!\sum_{k' \in \mathcal{K} \setminus \{k\}} \sum_{m' \in \mathcal{M}} \!\!\!\Big(\sigma_{m',s}^{k',k}[t] f_{m'k}\Big).
\end{equation} 
From \eqref{Y1_final} and \eqref{Y2_final}, we can finally rewrite $f_k^{\mathtt{tot}}[t]$  as shown in \eqref{reform_fk}. 

\textit{Derivation of \eqref{fc_total_reform}, \eqref{T_kc_ts_reform}, \eqref{T_kk_ts_reform} and \eqref{T_total__pro_reform}:} Combining \eqref{derive_Y1_p2}, \eqref{lemma2_eq3}, \eqref{lemma2_eq4} and \eqref{lemma3.2}, we can simplify $f_\mathtt{c}^{\mathtt{tot}}\left[ {{t}} \right],  T_{k}^{\mathtt{c}}[t_j],T_{k}^{k'}[t_j], T_{m,s}^{\mathtt{tot,pro}}\left[ {{t_j}} \right]$ as shown in \eqref{fc_total_reform}, \eqref{T_kc_ts_reform}, \eqref{T_kk_ts_reform} and \eqref{T_total__pro_reform}, respectively.

\textit{Derivation of \eqref{T_tot_cp_reform}:} First, we decompose $T_{m,s}^{\mathtt{tot,cp}}\left[ {{t_j}} \right]$ as follows:
$ T_{m,s}^{\mathtt{tot,cp}}\left[ {{t_j}} \right] = T_{m,s}^{\mathtt{ue,cp}}\left[ {{t_j}} \right] + {T_1} + {T_2} + {T_3}$, where
\begin{equation}
    \begin{array}{l}
{T_1} = \mathop {\max }\limits_{\forall k} \Big\{ {\begin{array}{*{20}{l}}
{T_{m,s}^{k,\mathtt{cp}}\left[ {{t_j}} \right]\sigma _{m,s}^k\left[ t \right]g_s^k\left[ t \right]}\\
{\times ( {1 - \sum_{k' \in \mathcal{K} \setminus \{k\}} {\sigma _{m,s}^{k,k'}\left[ t \right]}  - \sigma _{m,s}^{k,\mathtt{c}}\left[ t \right]} )}
\end{array}} \Big\}\\
{T_2} = \mathop {\max }\limits_{\forall k} \left\{ {\sigma _{m,s}^k\left[ t \right]\sum_{k' \in \mathcal{K} \setminus \{k\}} {\sigma _{m,s}^{k,k'}\left[ t \right]} g_s^{k'}\left[ t \right]T_{m,s}^{k',\mathtt{cp}}\left[ {{t_j}} \right]} \right\},\\
{T_3} = \mathop {\max }\limits_{\forall k} \left\{ {\sigma _{m,s}^k\left[ t \right]\sigma _{m,s}^{k,\mathtt{c}}\left[ t \right]T_{m,s}^{\mathtt{c,cp}}\left[ {{t_j}} \right]} \right\}.
\end{array}
\end{equation}
Following similar steps used in the transformation of $Y_1$, we apply \eqref{proof_logic_simplify} along with the conditions in \eqref{lemma1_eq1} and \eqref{lemma2_eq3} to decouple the strong coupling of binary variables. For $T_2$, we further break down this coupling using the conditions in \eqref{lemma2_eq3} and \eqref{lemma3.2}, combined with \eqref{proof_logic_simplify}. 
To simplify $T_3$, we apply condition \eqref{lemma2_eq4}. Finally, $ T_{m,s}^{\mathtt{tot,cp}}\left[ {{t_j}} \right]$ is rewritten as in \eqref{T_tot_cp_reform}.

\begingroup
\bibliographystyle{IEEEtran}
\bibliography{Bibliography}
\endgroup

\end{document}